\documentclass[reprint,superscriptaddress,amsmath,amssymb,aps,pra,floatfix]{revtex4-2}
\usepackage{graphicx} % Required for inserting images
\usepackage{xcolor}
\usepackage{etoolbox}
\usepackage{physics}
\usepackage{mathtools}
\usepackage{bbm}
\usepackage{amsthm}
\usepackage{hyperref}% add hypertext capabilities
\usepackage{comment}
\usepackage{enumerate}
\hypersetup{
    colorlinks,
    citecolor=blue,
    filecolor=blue,
    linkcolor=blue,
    urlcolor=blue
}

\DeclareMathAlphabet{\mymathbb}{U}{BOONDOX-ds}{m}{n}

\newtheorem{theorem}{Theorem}
\newtheorem{proposition}{Proposition}
\newtheorem{corollary}{Corollary}

\usepackage{cleveref}
\crefname{equation}{Eq.}{Eqs.}
\crefname{figure}{Fig.}{Figs.}
\crefname{section}{Sec.}{Secs.}
\crefname{theorem}{Theorem}{Theorems}
\crefname{corollary}{Corollary}{Corollaries}
\crefname{proposition}{Proposition}{Propositions}
\crefname{appendix}{Appendix}{Appendices}
\crefname{table}{Table}{Tables}

\newcommand{\iket}[1]{\lvert#1\rangle}

\newcommand{\braW}[1][M]{\bra{W_{\!#1}}}
\newcommand{\ketW}[1][M]{\ket{W_{\!#1}}}
\newcommand{\ketbraW}[1][M]{\ketbra{W_{\!#1}}}
\newcommand{\rhoW}[1][M]{\rho_{W_{\!#1},\eta}}
\NewDocumentCommand{\iketCat}{ O{M} O{\gamma} }{\iket{\operatorname{cat}_{\!#1}(#2)}}
\NewDocumentCommand{\ketCat}{ O{M} O{\gamma} }{\ket{\operatorname{cat}_{\!#1}(#2)}}
\newcommand{\rhoCat}[1][M]{\rho_{\operatorname{cat}_{\!#1}(\gamma),\eta}}

\allowdisplaybreaks
\begin{document}

\title{Witnessing genuine multipartite entanglement in phase space\texorpdfstring{\\}{ }with controlled Gaussian unitaries}

\author{Lin Htoo Zaw}
\affiliation{Centre for Quantum Technologies, National University of Singapore, 3 Science Drive 2, Singapore 117543}

\author{Jiajie Guo}
\email{jjguo@pku.edu.cn}
\affiliation{State Key Laboratory of Artificial Microstructure and Mesoscopic Physics, School of Physics, Frontiers Science Center for Nano-optoelectronics, \& Collaborative Innovation Center of Quantum Matter, Peking University, Beijing 100871, China}

\author{Qiongyi He}
\affiliation{State Key Laboratory of Artificial Microstructure and Mesoscopic Physics, School of Physics, Frontiers Science Center for Nano-optoelectronics, \& Collaborative Innovation Center of Quantum Matter, Peking University, Beijing 100871, China}
\affiliation{Collaborative Innovation Center of Extreme Optics, Shanxi University, Taiyuan, Shanxi 030006, China}
\affiliation{Hefei National Laboratory, Hefei 230088, China}

\author{Shuheng Liu}
\email{liushuheng@pku.edu.cn}
\affiliation{State Key Laboratory of Artificial Microstructure and Mesoscopic Physics, School of Physics, Frontiers Science Center for Nano-optoelectronics, \& Collaborative Innovation Center of Quantum Matter, Peking University, Beijing 100871, China}

\author{Matteo Fadel}
\email{fadelm@phys.ethz.ch}
\affiliation{Department of Physics, ETH Z\"{u}rich, 8093 Z\"{u}rich, Switzerland}

\begin{abstract}
Many existing genuine multipartite entanglement (GME) witnesses for continuous-variable (CV) quantum systems typically rely on quadrature measurements, which are challenging to implement in platforms where the CV degrees of freedom can be indirectly accessed only through qubit readouts.
In this work, we propose methods to implement GME witnesses through phase-space measurements in state-of-the-art experimental platforms, leveraging controlled Gaussian unitaries readily available in qubit-CV architectures.
Based on two theoretical results showing that sufficient Wigner negativity can certify GME, we present five concrete implementation schemes using controlled parity, displacement, and beam-splitter operations.
Our witnesses can detect paradigmatic GME states like the Dicke and multipartite $N00N$ states, which include the $W$ states as a special case, and GHZ-type entangled cat states.
We analyze the performance of these witnesses under realistic noise conditions and finite measurement resolution, showing their robustness to experimental imperfections.
Crucially, our implementations require exponentially fewer measurement settings than full tomography, with one scheme requiring only a single measurement on auxiliary modes.
The methods are readily applicable to circuit/cavity quantum electrodynamics, circuit quantum acoustodynamics, as well as trapped ions and atomic systems, where such dichotomic phase-space measurements are already routinely performed as native readouts.
\end{abstract}

\maketitle

\section{Introduction}
In a wide variety of quantum systems, such as cavity/circuit quantum-electrodynamics (cQED) \cite{cQED1,cQED2,cQED3,cQED4}, circuit quantum-acoustodynamics (cQAD) \cite{cQAD}, as well as trapped ions \cite{trappedIons1,trappedIons2} and atoms \cite{trappedAtoms}, the continuous-variable (CV) degrees of freedom are not directly characterized by quadrature measurements.
Instead, properties of the CV states are inferred by first coupling it to a qubit degree of freedom, for example through a Jaynes-Cummings interaction, then measuring the state of the qubit.
Thus, only dichotomic (i.e., binary) measurements are typically performed, like the expectation values of the parity using a controlled rotation, or the real or imaginary part of the displacement operator using a controlled displacement.
These correspond, respectively, to pointwise measurements of the Wigner and characteristic functions \cite{cQED-parity-1,cQED-parity-2,cQED-CD-1,cQED-CD-2,cQED-CD-3,cQAD}. In other words,  phase-space measurements are native in such systems.

In this work, we are interested in the detection of genuine multipartite entanglement (GME) in CV systems \cite{first-GME-paper}.
GME represents the strongest form of multipartite entanglement, characterized by quantum correlations that involve all parties of a system.
As a fundamental quantum resource, GME plays a crucial role in distributed quantum computing and quantum communication protocols \cite{GME-useful-distributed-computing,GME-useful-secret-sharing} and is also essential for achieving the ultimate sensitivity limit in quantum metrology \cite{TothPRA2012}.

As with other forms of quantum correlations, the detection of GME is typically achieved through witnesses, namely inequalities that are satisfied for all non-GME states.
Experimental measurements resulting in a violation of such inequalities then reveal the presence of GME.
Crucially, witnesses require significantly fewer measurements than full quantum state tomography \cite{chou2025deterministic,he2024efficient}, making them far more economical to implement in practice, while also introducing fewer systematic errors \cite{SchwemmerSystematic2015}.

Most existing CV multipartite entanglement witnesses, however, rely on quadrature measurements \cite{shalm2013three,ZhangGenuine2023,HyllusOptimal2006,TehPRA14,TehPRA22,Leskovjanova2025minimalcriteria,ToscanoSystematic2015}, which makes them particularly suited to optical platforms.
In contrast, entanglement witnesses relying on binary phase-space measurements have only been derived for bipartite scenarios \cite{DetectingZhang2013,GittsovichNonclassicality2015,GholipourShahandeh2016,JayachandranDynamics2023,ZawCertifiable2024,liu2024quantumentanglementphasespace}.
A practical generalization of the latter criteria to multipartite cases is highly valuable but still missing from the existing literature.

In this work, we fill this gap by proposing experimentally feasible implementations of GME witnesses based on phase-space measurements in state-of-the-art experimental platforms, which require exponentially fewer measurement settings than full state tomography.
For this, we build upon two primary theoretical results from our companion Letter \cite{CV-GME-letter}, where we demonstrate complementary ways in which Wigner negativities can reveal GME.
Specifically, GME states can exhibit sufficiently large negativity volumes in certain phase-space slices, exceeding the bounds allowed for biseparable states, and can feature negativities sharp enough to withstand smoothing operations that would completely remove them in all biseparable states.

Based on these principles, we present five concrete schemes implementing GME witnesses using conditional Gaussian unitaries that are accessible in current qubit-CV architectures.
We analyze the performances of these witnesses for detecting paradigmatic examples of GME states under realistic experimental imperfections, including the effects of noise, losses, and finite measurement resolution.
Our results show how different criteria offer complementary advantages in terms of measurement overhead, noise resilience, and detectable states.

\section{Background and Motivation}

Let us begin by recalling the theoretical background of phase-space quantum information, genuine multipartite entanglement, and the standard operations and measurements available in the experimental platforms under consideration.

\subsection{\label{sec:phase-space-QI}Phase-space quantum information}

We consider CV systems comprising $M$ modes, which are specified by the annihilation operators $\vec{a} \coloneqq (a_1,a_2,\dots,a_M)$ that satisfy the canonical commutation relations
\begin{equation}\label{eq:CCR}
    \comm{a_m}{a_{m'}} = \comm{a_m^\dag}{a_{m'}^\dag} = 0, \qquad
    \comm{a_m}{a_{m'}^\dagger} = \mathbbm{1} \delta_{m,m'} .
\end{equation}

In CV quantum information, states of this multimode system are conveniently represented in phase space, namely the 2$M$-dimensional space specified by the position and momentum operators associated to each mode.

To work in this phase-space picture, we introduce the parity $\Pi_m \coloneqq e^{-i\pi a_m^\dag a_m}$ and displacement $D_m(\alpha_m) \coloneqq e^{ \alpha_m a_m^\dagger - \alpha_m^* a_m}$ operators defined on the $m$th mode.
They act on the $m$th annihilation operator as $\Pi_m^\dag a_m \Pi_m = -a_m$ and $D_m^\dag (\alpha_m) a_m D_m (\alpha_m) = a_m + \alpha_m$.
We can further specify corresponding operators that act on all modes as $\Pi \coloneqq \prod_m \Pi_m$ and $D(\vec{\alpha}) \coloneqq \prod_m D_m(\alpha_m)$, where $\vec{\alpha} = (\alpha_1,\alpha_2,\dots,\alpha_M)$.

Composing these operators gives us the displaced parity operator $\Pi(\vec{\alpha}) \coloneqq D(\vec{\alpha}) \Pi D^\dag(\vec{\alpha})$, which allows us to define the Wigner function of a state $\rho$ as
\begin{equation}
    W_{\rho}(\vec{\alpha}) \coloneqq \pqty{\frac{2}{\pi}}^{M}\tr[\rho \Pi(\vec{\alpha})].
\end{equation}
Here $\vec{\alpha}$ are phase-space coordinates such that $\Re(\alpha_m)$ is proportional to the position and $\Im(\alpha_m)$ the momentum of the $m$th mode.
The Wigner function is a quasiprobability distribution in phase space:
It has the property of a joint distribution of $\vec{\alpha}$ in that its marginals are probability distributions of the system's position and momentum, except that its value can be negative at some points \cite{quasiprobability-review}.
The presence of negativities in $W_\rho(\vec{\alpha})$ is one notion of nonclassicality, as it demonstrates that the observed behavior cannot be simulated by a joint classical probability distribution for $\vec{\alpha}$.

A related function is the characteristic function of the state $\rho$, which is defined as the Fourier transform of the Wigner function
\begin{equation}
    \chi_\rho(\vec{\xi}) \coloneqq \int_{\mathbb{C}^M}\dd[2M]{\vec{\alpha}}\; e^{\vec{\xi}\wedge\vec{\alpha}}\;
    W_\rho(\vec{\alpha}) = \tr[\rho D(\vec{\xi})].
\end{equation}
Here, $\vec{A} \wedge \vec{B} \coloneqq \sum_m A_m B_m^\dag - B_m A_m^\dag$ is the wedge product.
Importantly, both the Wigner and characteristic functions are tomographically complete in that they provide a full description of the state of the system equivalent to the information contained in its density matrix.

\subsection{Genuine multipartite entanglement}
Entanglement is a form of nonclassical correlation found in multipartite quantum systems.
Given a bipartition of the system, states that cannot be prepared using only local operations and classical communication between the parties are called entangled; otherwise, they are said to be (bi-)separable.

We focus here on GME, the strongest form of entanglement that can be present in a multipartite system.
A state is GME if it is not a statistical mixture of biseparable states.
More precisely for $M$-mode systems, $\rho$ is GME if
\begin{equation}
    \rho \neq \sum_{(\mathcal{A} \mid \bar{\mathcal{A}})} \sum_{k} p_{\mathcal{A}}^{(k)} \rho_{\mathcal{A}}^{(k)} \otimes \rho_{\bar{\mathcal{A}}}^{(k)} ,
\end{equation}
where $p_{\mathcal{A}}^{(k)} \geq 0$ is a probability distribution satisfying $\sum_{(\mathcal{A} \mid \bar{\mathcal{A}})} \sum_{k} p_{\mathcal{A}}^{(k)}=1$, the notation $(\mathcal{A}|\bar{\mathcal{A}})$ indicates  bipartitions of the modes into $\mathcal{A}=\{m_n\}_{n=1}^N$ and $\bar{\mathcal{A}}=\{m\}_{m=1}^M \setminus \mathcal{A}$ for $1\leq N<M$, and $\rho_{\{m_1, m_2, \ldots, m_N\}}^{(k)}$ are states defined locally on the $\{a_{m_1}, a_{m_2}, \ldots, a_{m_N}\}$ modes.

\subsection{\label{CV-qubit-readout}Continuous variable systems with qubit readouts}
In this work, we consider experimental platforms where the continuous variable degrees of freedom under investigation are measured through a two-level (qubit) system.
Such platforms include cQED and cQAD devices, as well as trapped ions and atoms systems.
In such architectures, measurements of the CV state are performed by first coupling the CV and qubit degrees of freedom, then reading out the qubit state to infer properties about the CV state.

Let us say that we wish to measure an observable on the CV degree of freedom, which we present as a quantum circuit in \cref{fig:cU-circuit}(a).
A common coupling choice is to perform a Gaussian unitary on a CV mode conditioned on the qubit degree of freedom.
For a unitary $U$, its controlled version is $cU \coloneqq \ketbra{\uparrow}\otimes\mathbbm{1} + \ketbra{\downarrow}\otimes U$.
Its representation as a circuit is shown in \cref{fig:cU-circuit}(b).
Note that this definition deviates slightly from the convention $\ketbra{\uparrow}{\downarrow}\otimes\mathbbm{1} + \ketbra{\downarrow}{\uparrow}\otimes U$ used in previous cQED literature, but they are equivalent up to a local qubit phase that can be reintroduced in postprocessing.

\begin{figure*}
    \centering
    \includegraphics{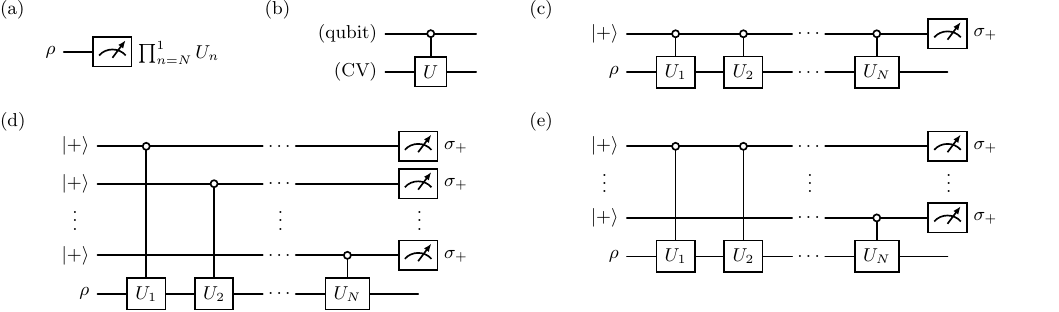}
    \caption{\label{fig:cU-circuit}
        Consider that like in (a), we wish to measure ${\tr}(\rho\prod_{n=N}^1 U_{n})$ with respect to the CV state $\rho$ using only qubit readouts.
        This can be implemented using (b), the controlled unitary $cU = \ketbra{\uparrow}\otimes\mathbbm{1} + \ketbra{\downarrow}\otimes U$.
        Then, ${\tr}(\rho \prod_{n=N}^1 U_{n})$ can be measured by first preparing all the readout qubits as $\ket{+}$, then either (c) performing all controlled unitaries on one readout qubit and measuring $\ev{\sigma_+}$, (d) performing them each on their own readout qubit and measuring $\ev{\sigma_{+}\otimes\sigma_{+}\otimes\cdots\otimes\sigma_{+}}$, or (e) some combination of both.
    }
\end{figure*}

Now, given the circuit in \cref{fig:cU-circuit}(c), which consists of initializing the system in the state $\ketbra{+}\otimes\rho$ with $\ket{\pm} \coloneqq (\ket{\uparrow}\pm \ket{\downarrow})/\sqrt{2}$, performing a sequence of controlled unitary gates $cU_1,cU_2,\dots$, before measuring the qubit.
Then the expectation value of $\sigma_x \coloneqq \ketbra{\uparrow}{\downarrow} + \ketbra{\downarrow}{\uparrow}$ gives
\begin{equation}
\begin{aligned}
\ev{\sigma_x} &= \tr[ \sigma_x cU_N \cdots cU_1 (\ketbra{+}\otimes\rho) cU_1^\dag \cdots cU_N^\dag ] \\
&= \Re[\tr\pqty\big{\rho \underbrace{U_{N} U_{N-1}\cdots U_1}_{=\prod_{n=N}^1 U_n})}],
\end{aligned}
\end{equation}
while the same for $\sigma_y \coloneqq -i\ketbra{\uparrow}{\downarrow} + i\ketbra{\downarrow}{\uparrow}$ gives $\ev{\sigma_y} = \Im[{\tr}(\rho \prod_{n=N}^1 U_{n})]$.
Hence, the readout of $\sigma_+ \coloneqq \sigma_x + i\sigma_y$ in the qubit degree of freedom gives us the expectation $\ev{\sigma_+} = {\tr}(\rho \prod_{n=N}^1 U_{n} )$ in the CV degree of freedom.

Depending on the physical layout and experimental limitations of the qubit-CV system, not all qubit and CV degrees-of-freedom might be connected, and it could also be the case that different qubit auxiliaries are required to perform different controlled unitaries.
Then, similar circuits can perform the same measurement using a different auxiliary qubit for each unitary instead of a single auxiliary qubit performing all of them, as shown in \cref{fig:cU-circuit}(d), while hybrids of both scenarios are also possible, as shown in \cref{fig:cU-circuit}(e).

In constructing our GME witnesses, we will use controlled versions of particular Gaussian unitaries.
Gaussian unitaries are those generated by Hamiltonians that are at most quadratic in the creation and annihilation operators.
In particular, we will primarily use the following subset of unitaries:

(1) Parity operators $\Pi_m = e^{-i\pi a_m^\dag a_m}$ and $\Pi = \prod_{m=1}^M \Pi_m = e^{-i\pi \vec{a}^\dag \vec{a}}$, with the actions $\Pi_m a_{m'\neq m} \Pi_m = a_{m'}$, $\Pi_m a_m \Pi_m = -a_m$, and $\Pi \vec{a} \Pi = -\vec{a}$.

(2) Displacements $D(\vec{\alpha}) = e^{ \vec{\alpha} \wedge \vec{a}}$, with the action $D^\dagger(\vec{\alpha})\vec{a}D(\vec{\alpha}) = \vec{a} + \vec{\alpha}$.

(3) Multiport beam splitters $e^{\vec{a}^\dag(\ln\mathbf{U})\vec{a}}$, where $\mathbf{U}$ is an $M \times M$ unitary matrix, with the action $e^{-\vec{a}^\dag(\ln\mathbf{U})\vec{a}}\,\vec{a}\,e^{\vec{a}^\dag(\ln\mathbf{U})\vec{a}} = \mathbf{U}\vec{a}$

One reason for focusing on these unitaries is that they have been demonstrated experimentally as controlled unitaries in many qubit-CV architectures.
We shall also study these three operations separately as primary building blocks for more complex Gaussian unitaries, instead of studying Gaussian unitaries in all generality.
This is because the required experimental parameters to implement the operations are generally quite different from one another, and thus a particular experimental setup might only be able to perform a strict subset of those operations.

Another feature of these controlled Gaussian unitaries is their relationships to some important quantities in the phase-space picture of CV quantum information.
As partly mentioned in \cref{sec:phase-space-QI}, controlled parities and multiport beam splitters provide direct measurements of the Wigner quasiprobability distribution, while controlled displacement operators provide direct measurements of its Fourier transform, the characteristic function.

\section{Summary of Theoretical Prerequisites}
\subsection{Primary theorems}
For convenience, we restate here some key theoretical results found in our companion Letter \cite{CV-GME-letter}, upon which we will base our additional results in the rest of the work.

The first is a corollary of Theorem 1 from our companion Letter that relates the integral of a Wigner function over a finite region to GME. \stepcounter{theorem}
\begin{corollary}[From Ref.~\cite{CV-GME-letter}: GME criterion with Wigner function measurements over a finite region]\label{col:GME-Wigner-witness}
    Let the absolute volume of the Wigner function on $\{\alpha\vec{y} + \alpha^*\vec{z} : \alpha \in \omega\}$, where $\vec{y}\circ\vec{y}^*-\vec{z}\circ\vec{z}^* = \vec{1}$ and $\omega \subseteq \mathbb{C}$ is Lebesgue-measurable, be
    \begin{equation}\label{eq:GMEWignerFiniteRegionDef}
    \begin{aligned}
        \mathcal{V}_{2D}(\rho;\omega) &\coloneqq \pqty{\frac{\pi}{2}}^{M-1}\int_{\omega}\dd[2]{\alpha} \abs\big{W_{\rho}\pqty{\alpha\vec{y} + \alpha^*\vec{z}}}.
    \end{aligned}
    \end{equation}
    Then, $\mathcal{V}_{2D}(\rho;\omega) > (2\sqrt{M-1})^{-1}$ implies that $\rho$ is GME.
\end{corollary}
Here, $[\mathbf{A}\circ\mathbf{B}]_{m,n} = [\mathbf{A}]_{m,n}[\mathbf{B}]_{m,n}$ is the element-wise product and $\vec{1} = (1,1,\dots,1)$ is a vector of ones.
The quantity $\mathcal{V}_{2D}$ quantifies the amount of negativity present along the two-dimensional slice $\{\alpha\vec{y} + \alpha^*\vec{z} : \alpha \in \omega\}$ of phase space, and it was also shown in the Letter that $\mathcal{V}_{2D} \leq (2\sqrt{M-1})^{-1}$ for all Wigner positive states.

The next theorem concerns the Wigner function of the center-of-mass mode $a_{+}$, which is defined as
\begin{equation}
    a_+ \coloneqq \frac{\vec{y}^T\vec{a} + \vec{a}^\dag\vec{z}}{\sqrt{M}} = \frac{1}{\sqrt{M}}\sum_{m=1}^{M}\pqty{y_m a_m + z_m a_m^\dag},
\end{equation}
where $\vec{y},\vec{z}\in\mathbb{C}^M$ with $\vec{y}\circ\vec{y}^*-\vec{z}\circ\vec{z}^* = \vec{1}$ as before.
If left unspecified, take $\vec{y} = \vec{1}$ and $\vec{z} = (0,\dots,0)$.
Further specifying $M-1$ relative modes $\{a_{-m}\}_{m=2}^{M}$ such that $\{a_+\}\cup\{a_{-m}\}_{m=2}^{M}$ satisfies the canonical commutation relations, the reduced state $\tr_-\rho$ that describes only the center-of-mass motion of the system is then
\begin{equation}
    \tr_-\rho \coloneqq \tr_{a_{-2}}\tr_{a_{-3}}\cdots\tr_{a_{-M}}\rho.
\end{equation}
With these definitions, we can now state the theorem.

\begin{theorem}[From Ref.~\cite{CV-GME-letter}: Negativity of the smoothed Wigner function of the center of mass implies GME]\label{thm:GMN-reduced-state}
    Choose $M-2$ states $\mathcal{R} = \{\varrho_m\}_{m=1}^{M-2}$. Define the smoothed Wigner function of the center of mass of the system as
\begin{equation}
\widetilde{W}_{\tr_{-}\rho}(\alpha;\mathcal{R})\\
    \coloneqq \int_{\mathbb{C}}\dd[2]{\beta}
    \;W_{\tr_{-}\rho}(\beta)
    \;K\pqty{\alpha-\beta;\mathcal{R}},
\end{equation}
where $K(\alpha,\mathcal{R})$ is the convolution kernel
\begin{equation}
\begin{aligned}
    K(\alpha;\mathcal{R})
    &\coloneqq
    \int_{\mathbb{C}^{M-2}}\dd[2(M-2)]{\vec{\gamma}}
    \;\prod_{m=1}^{M-2}W_{\varrho_m}(\gamma_m)\\[-1ex]
    &\hspace{4.5em}{}\times{}2\pqty{1-M^{-1}}
    \;\delta\pqty{
        \alpha-\tfrac{\vec{\gamma}^T\vec{1}}{\sqrt{M}}
    }.
\end{aligned}
\end{equation}
Then, the smoothed Wigner function lower bounds, up to a factor, the trace distance to all non-GME states as
\begin{equation}
\begin{aligned}
    \max\Bqty{0,-\widetilde{W}_{\tr_{-}\rho}(\alpha;\mathcal{R})} \leq \frac{2}{\pi}\min_{\sigma \notin \operatorname{GME}} \| \sigma - \rho \|_1.
\end{aligned}
\end{equation}
Hence, $\exists\alpha : \widetilde{W}_{\tr_{-}\rho}(\alpha;\mathcal{R}) < 0$ implies that $\rho$ is GME.
\end{theorem}
Here, $\| A \|_1$ is the trace norm, or sum of the singular values, of $A$. Choosing $\mathcal{R}=\mathcal{R}_G$ to be Gaussian states gives a simple expression for the kernel function
\begin{equation}\label{eq:filter-function-Gaussian}
\begin{aligned}
    K(\alpha;\mathcal{R}_G) &= \frac{1-M^{-1}}{\pi\sqrt{\det\Sigma'}} e^{-\frac{1}{2}\abs{{\Sigma'}^{-\frac{1}{2}}\spmqty{
        \Re[\alpha-\alpha']\\
        \Im[\alpha-\alpha']
    }}^2} \\
    & \qquad\qquad{}\text{such that }{}\sqrt{\det\Sigma'} \geq \frac{M-2}{4M},
\end{aligned}
\end{equation}
where $\alpha' \in \mathbb{C}$ and $\Sigma' = \Sigma^{\prime T} \in \mathbbm{R}^{2\times2}$ are, respectively, related to the means and covariances of $\mathcal{R}_G$.
This implies the following corollary as a special case of \cref{thm:GMN-reduced-state}, which relates to another quantifier of nonclassicality previously defined in the literature.
\begin{corollary}[From Ref.~\cite{CV-GME-letter}: Enough nonclassicality depth of the center of mass implies GME]\label{col:nonclassicality-depth-reduced-state}
    The nonclassicality depth $\tau_c$ of a state $\rho$ is defined as \cite{nonclassicality-depth}
    \begin{equation}
        \tau_c(\rho) \coloneqq \min\!\Bqty{\tau : \forall \alpha : \frac{1}{\pi\tau}\int_{\mathbb{C}}\dd[2]{\beta}
        P_\rho(\beta)e^{-\frac{\abs{\alpha-\beta}^2}{\tau}} \geq 0 },
    \end{equation}
    where $P_\rho(\alpha)$ is the Glauber $P$ function of $\rho$ such that $\rho = \int_{\mathbb{C}}\dd[2]{\alpha} P_{\rho}(\alpha)\ketbra{\alpha}$, and $\ket{\alpha}$ is the coherent state.
    Then, $\tau_c(\tr_{-}\rho) > 1 - M^{-1}$
    implies that $\rho$ is GME.
\end{corollary}

\subsection{Relation to Gaussian operations}
In this section, we report some relationships between Gaussian operations and phase-space quantities.
Since our witnesses depend on measurements of Wigner or characteristic functions, these results will be useful for constructing multipartite entanglement witnesses in terms of (controlled-)Gaussian operations on the CV degrees of freedom.

The first useful result, proven in Appendix~\ref{apd:bochner-extended}, connects the expectation value of displacement operators to that of the displaced parity.
\begin{proposition}[Bochner's theorem and bounds of absolute volume of displaced parity]\label{prop:bochner-extended}
    For any self-adjoint trace-class operator $R = R^\dag$ such that $\left\|R\right\|_1 < \infty$, consider a collection of phase-space vectors $\Xi = \{\vec{\xi}_n\}_{n=1}^N$, and the matrix $\mathbf{C}(R;\Xi) \in \mathbb{C}^{N \times N}$ constructed as
    \begin{equation}
        [\mathbf{C}(R;\Xi)]_{n,n'} \coloneqq \frac{1}{N}\tr[R D(\vec{\xi}_n-\vec{\xi}_{n'})].
    \end{equation}
    Then, the trace norm of $\mathbf{C}(R;\Xi)$ lower bounds the absolute integral of the expectation value of the displaced parity as
    \begin{equation}
        \left\|\mathbf{C}(R;\Xi)\right\|_1 \leq \pqty{\frac{2}{\pi}}^M\int_{\mathbb{C}^M}\dd[2M]{\vec{\alpha}} \abs{\tr[ R \Pi(\vec{\alpha})]}.
    \end{equation}
\end{proposition}

Another useful relation, proven in Appendix~\ref{apd:center-of-mass-parity-is-beam splitter}, is the connection between the parity on collective modes (like the center of mass or relative modes) and multiport beam splitters.
\begin{proposition}[Parity operations on collective modes as multiport beam splitters]\label{prop:center-of-mass-parity-is-beam splitter}
    The parity operators $\Pi_{\pm M}$, where $\Pi_{+M} \coloneqq e^{i\pi a_+^\dag a_+}$ for $\sqrt{M}a_+ = \sum_{m=1}^M a_m$ and $\Pi_{-M} \coloneqq e^{i\pi (\vec{a}^\dag\vec{a} - a_+^\dag a_+)} = \Pi\,\Pi_{+M}$, are multiport beam splitters that act on the local modes as
    \begin{equation}\label{eq:VpmMDef}
        \Pi_{\pm M} a_m \Pi_{\pm M} = \pm \pqty{1-\frac{2}{M}}a_m \mp \frac{2}{M}\sum_{m'\neq m} a_{m'}.
    \end{equation}
\end{proposition}
In the following sections, we will use the above results to formulate concrete schemes that experimentally implement our phase-space GME witnesses.

\section{Implementation of GME Witnesses}
\begin{table*}[t!]
\renewcommand{\arraystretch}{3}
\setlength{\tabcolsep}{12pt}
\hspace*{-2.5em}\begin{tabular}{ r c c c c }
\cline{2-5}
& \textbf{Based on} & \textbf{Required $cU$} & \begin{tabular}{@{}c@{}}\\[-10ex]\textbf{Additional}\\[-5ex]\textbf{requirements}\end{tabular} & \textbf{GME condition} \\[1ex]
\cline{2-5}
\ref{wit:P}
    & Corollary~\ref{col:GME-Wigner-witness}
    & Parity
    & Discretization
    & $\begin{array}{c}
        \int_{\omega}\dd[2]{\alpha}\abs{
            \tr[\rho\, \Pi(\alpha\vec{y}+\alpha^*\vec{z})]
        } > \frac{\pi}{4\sqrt{M-1}} \\[-4ex]
        \text{where $\omega\subseteq\mathbb{C}$ is Lebesgue-measurable}
        \end{array}$\\
\ref{wit:D-BS}
    & Corollary~\ref{col:GME-Wigner-witness}
    & \begin{tabular}{@{}c@{}}\\[-10ex]\text{Displacement \&}\\[-5ex]\text{beam splitter}\end{tabular}
    & ---
    & $\begin{array}{c}
        \left\|\mathbf{C}\right\|_1 > \frac{M}{2\sqrt{M-1}}\text{ where }\mathbf{C} \in \mathbb{C}^{N \times N},\\[-3.5ex]
    [\mathbf{C}]_{n,n'} = \frac{1}{N}\tr[
        \rho \,
        \Pi_{-M} \,
        D\pqty\big{(\xi_n-\xi_{n'})\vec{1}}
    ]
    \end{array}$ \\
\ref{wit:BS-A}
    & Theorem~\ref{thm:GMN-reduced-state}
    & Beam splitter
    & Ancilla modes
    & $ \tr[\pqty{\rho \otimes \bigotimes_{m=1}^{M-2}\varrho_{m}} \Pi_{+(2M-2)}] < 0 $ \\
\ref{wit:BS-Rand}
    & Theorem~\ref{thm:GMN-reduced-state}
    & Beam splitter
    & \begin{tabular}{@{}c@{}}\\[-10ex]\text{Random}\\[-5ex]\text{displacements}\end{tabular}
    & $\begin{array}{c}
        \ev{\tr[
            D(\hat{\beta}\vec{1})
            \rho
            D^\dag(\hat{\beta}\vec{1}) \,
            \Pi_{+M}
        ]} < 0,\\[-3.5ex]
    \text{where }\spmqty{
            \Re[\hat{\beta}] \\
            \Im[\hat{\beta}]
        } \sim \mathcal{N}\pqty{
            -\frac{1}{\sqrt{M}}\spmqty{
                \Re[\alpha] \\
                \Im[\alpha]
            },
            \Sigma
        },\\[-3ex]
        \sqrt{\det\Sigma} \geq \frac{M-2}{4M^2}
    \end{array}$ \\
\ref{wit:D}
    & Theorem~\ref{thm:GMN-reduced-state}
    & Displacement
    & ---
    & $\begin{array}{c}
    \left\|\mathbf{C}\circ\mathbf{K}\right\|_1 > 1\;\text{where}\;
    \mathbf{C},\mathbf{K} \in \mathbb{C}^{N \times N},\\[-3.5ex]
    [\mathbf{C}]_{n,n'} = \frac{1}{N}\tr[
        \rho \,
        D\pqty\big{
            (\xi_n-\xi_{n'})\vec{y} +
            (\xi_n^*-\xi_{n'}^*)\vec{z}
        }
    ],\\[-3.5ex]
    [\mathbf{K}]_{n,n'} = \prod_{m=1}^{M-2} \chi_{\varrho_m}(\xi_n-\xi_{n'}) 
    \end{array}$  \\[1ex]
\cline{2-5}
\end{tabular}
\caption{
    \label{table:ImplementationSummary}
    A summary of implementation of our GME witnesses.
    Here $\vec{y},\vec{z}\in\mathbb{C}^M$ such that $\vec{y}\circ\vec{y}^* - \vec{z}\circ\vec{z}^* = \vec{1}$, $\{\xi_{n}\}_{n=1}^N$ is any choice of $N$ phase-space points, and $\mathcal{N}(\vec{\mu},\Sigma)$ is the multivariate normal distribution with mean $\vec{\mu}$ and covariance matrix $\Sigma$.
}
\end{table*}

Let us now demonstrate that our GME witnesses can be implemented in state-of-the-art experimental platforms through convenient quantum circuits.
The types of measurements we will consider are
\begin{enumerate}
    \item Witness in Sec.~\ref{wit:P}: Wigner function measurements $W_\rho(\vec{\alpha})$, performed locally;
    \item Witnesses in Sec.~\ref{wit:D-BS}~and~\ref{wit:D}: characteristic function measurements $\chi_\rho(\vec{\xi})$, performed locally;
    \item Witnesses in Sec.~\ref{wit:D-BS},~\ref{wit:BS-A},~and~\ref{wit:BS-Rand}: Wigner function measurements of center of mass $W_{\tr_{-}\rho}(\alpha)$ or relative modes, performed collectively.
\end{enumerate}

Local Wigner and characteristic function measurements are routinely performed via controlled displaced-parity and controlled displacement operations \cite{cQED-parity-1,cQED-parity-2,cQED-CD-1,cQED-CD-2,cQED-CD-3,cQAD}.

Meanwhile, measurements of the collective Wigner function of the center-of-mass or relative modes can be implemented by either directly coupling the qubit to a collective mode \cite{trapped-ion-center-of-mass}, or by implementing $\Pi_{\pm M}$ from \cref{eq:VpmMDef} as a multimode beam splitter.
We will use the latter implementation in the rest of the work due to its wider availability and familiarity: Qubit-controlled beam-splitter operations between two modes have been experimentally implemented in different qubit-CV architectures \cite{conditional-beam-splitter-1,conditional-beam-splitter-2}, while qubit-controlled operations of arbitrary multimode beam splitters can be implemented using a sequence of two-mode beam splitters \cite{multiport-beam-splitter-decomposition-1,multiport-beam-splitter-decomposition-2}.

A summary of all implementations of our witnesses is given in \cref{table:ImplementationSummary}.
Importantly, all of our witnesses are significantly less demanding than full state tomography: They only require measurements either over a finite region, or at a finite number of points, in phase space.
Most remarkably, with the help of auxiliary modes, a Wigner function measurement at a single point in phase space is sufficient to certify GME.

\subsection{\label{wit:P}With controlled parity}

A straightforward implementation of \cref{col:GME-Wigner-witness} involves measuring the controlled displaced-parity operator $\ketbra{\uparrow}\otimes\mathbbm{1} + \ketbra{\downarrow}\otimes \Pi_m(\alpha_m)$ on each mode for a chosen two-dimensional slice in phase space.
A circuit of the required measurement is given in \cref{fig:wit:P}.
There, we use the convention from \cref{fig:cU-circuit} where only $\Pi(\alpha)$ measurements on each mode are illustrated, which is actually implemented through a controlled displaced-parity measurement followed by a qubit readout.

With this measurement scheme, we can implement the GME witness with the following steps.
\begin{enumerate}
    \item Choose coefficients $\vec{y},\vec{z} \in \mathbb{C}^M : \vec{y}\circ\vec{y}^* - \vec{z}\circ\vec{z}^* = \vec{1}$
    \item Using the circuit in \cref{fig:wit:P}, measure $\langle \Pi(\alpha\vec{y}+\alpha^*\vec{z}) \rangle$ over a finite, Lebesgue-measurable two-dimensional slice $\alpha \in \omega \subseteq\mathbb{C}$ of phase space
    \item Then, integrate the absolute values of the measurement outcomes over this slice. If
    \begin{equation}\label{eq:wit:P}
        \int_{\omega}\dd[2]{\alpha} \abs{\ev{\Pi(\alpha\vec{y}+\alpha^*\vec{z})}} > \frac{\pi}{4\sqrt{M-1}},
    \end{equation}
    then the GME of the input state is certified
\end{enumerate}
The last step is simply a direct implementation of \cref{col:GME-Wigner-witness} with $\Pi(\vec{\alpha}) = (\pi/2)^{M}W_\rho(\vec{\alpha})$.

In principle, this only requires an integral over a finite region on a two-dimensional slice, which is much smaller than the $2M$ dimensions of the entirety of phase space, and thus already an exponential improvement over tomography.
In practice, \cref{eq:wit:P} is computed by discretizing and performing numerical integration over a finite number of representative points.
The effects of discretization, including some rigorous and heuristic approaches to error analysis, will be discussed in \cref{sec:P-finite}.

\begin{figure}
    \centering
    \includegraphics{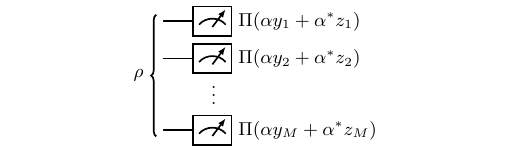}
    \caption{\label{fig:wit:P}
        Witness in Sec.~\ref{wit:P}: Implementation of \cref{col:GME-Wigner-witness} through local parity measurements on each mode.
        The measurements are performed for values of $\alpha$ that cover a finite region $\alpha \in \omega\subseteq\mathbb{C}$.
    }
\end{figure}

\subsection{\label{wit:D-BS}With controlled displacement and beam splitters}
Another implementation of \cref{col:GME-Wigner-witness} requires the measurement scheme as illustrated in \cref{fig:wit:D-BS}.
Here we first apply the controlled multiport beam splitter $\ketbra{\uparrow}\otimes\mathbbm{1} + \ketbra{\downarrow}\otimes \Pi_{-M}$ as defined in \cref{eq:VpmMDef} to all modes of $\rho$ and then measure $\sigma_x$ on the qubit and perform displacement measurements on each mode.
This circuit implements the measurement
\begin{equation}
\begin{aligned}
    \ev{\sigma_x \otimes  D(\xi\vec{1})}
    &= \tr[\rho \; D_{+}(\sqrt{M}\xi) \; \Pi_{-M}],
\end{aligned}
\end{equation}
where $D_+(\xi_+) \coloneqq e^{\xi_+ \wedge a_+}$ is the displacement operator on the center of mass $\sqrt{M} a_+ = \sum_{m=1}^M a_m$.
Thus, this is a hybrid phase-space measurement, where we measure the characteristic function on the center of mass, but the Wigner function on the relative modes.

\begin{figure}
    \centering
    \includegraphics{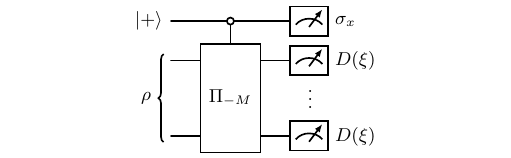}
    \caption{
        \label{fig:wit:D-BS}
        Witness in Sec.~\ref{wit:D-BS}: Implementation of \cref{col:GME-Wigner-witness} using a controlled beam splitter and local displacement measurements on each mode.
        The measurements are performed at a finite number of points $\xi \in \{\xi_n-\xi_{n'}\}_{n \leq n'}$.
    }
\end{figure}

Using this measurement scheme, the GME witness can be implemented as follows.
\begin{enumerate}
    \item Choose a set of phase-space points $\Xi = \{\xi_n\}_{n=1}^N$
    \item Measure $\langle \sigma_x\otimes D((\xi_n-\xi_{n'})\vec{1})\rangle$ using the circuit in \cref{fig:wit:D-BS} at $N(N+1)/2$ pairwise differences $\{\xi_n-\xi_{n'}:1 \leq n \leq n' \leq N\}$
    \item Collect the measurement outcomes into a finite matrix $\mathbf{C}(\rho;\Xi) \in \mathbb{C}^{N \times N}$ with entries $[\mathbf{C}(\rho;\Xi)]_{n,n'} = \frac{1}{N}\langle{\sigma_x \otimes  \bigotimes_{m=1}^M D_m(\xi_n - \xi_{n'})}\rangle$ for $n \leq n'$, while $\mathbf{C} = \mathbf{C}^\dagger$ is imposed to obtain the missing terms
    \item Finally, if $\left\|\mathbf{C}(\rho;\Xi)\right\|_1 > M(2\sqrt{M-1})^{-1}$, then the input state is GME
\end{enumerate}
The final step is a consequence from
\cref{prop:bochner-extended}, where
\begin{equation}
\begin{aligned}
    \left\|\mathbf{C}(\rho;\Xi)\right\|_1 &\leq \frac{2}{\pi}\int_{\mathbb{C}}\dd[2]{\alpha}\abs\big{
        \tr[\rho\Pi_{+M}(\alpha)\Pi_{-M}]
    } \\
    &= \pqty{\frac{\pi}{2}}^{M-1}\int_{\mathbb{C}}\dd[2]{\alpha}\abs{W_\rho\pqty{\alpha\vec{1}/\sqrt{M}}} \\
    &= M\pqty{\frac{\pi}{2}}^{M-1}\int_{\mathbb{C}}\dd[2]{\alpha}\abs{W_\rho(\alpha\vec{1})},
\end{aligned}
\end{equation}
which, together with \cref{col:GME-Wigner-witness}, implies the above condition.

\subsection{\label{wit:BS-A}With controlled beam splitter and ancilla modes}
Next we turn to an implementation of \cref{thm:GMN-reduced-state} that directly measures the smoothed Wigner function $\widetilde{W}_{\tr_{-}\rho}(\alpha;\mathcal{R})$.
While its measurement overhead is minimized compared to the previous methods, as this witness requires only a single measurement setting, it introduces $M-2$ auxiliary modes as a trade-off.
The required operations are illustrated as a circuit in \cref{fig:wit:BS-A}, and are performed as follows
\begin{enumerate}
    \item Prepare the auxiliary modes in a separable state $\bigotimes_{m=1}^{M-2}\varrho_m$
    \item Subject the joint $(M+(M-2))$-mode system to the controlled multiport beam splitter $\ketbra{\uparrow}\otimes\mathbbm{1} + \ketbra{\downarrow}\otimes \Pi_{+(2 M-2)}$ as defined in \cref{eq:VpmMDef}
    \item Measure $\sigma_x$ on the qubit. If $\langle\sigma_x\rangle < 0$, then the input state is GME
\end{enumerate}

\begin{figure}
    \centering
    \includegraphics{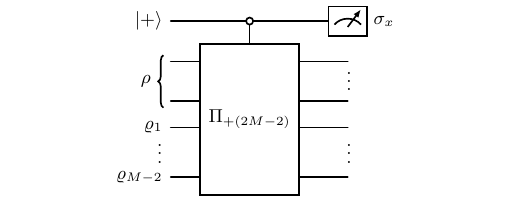}
    \caption{
        \label{fig:wit:BS-A}
        Witness in Sec.~\ref{wit:BS-A}: Implementation with a controlled beam splitter and auxiliary modes.
        Only a single measurement setting is needed.
    }
\end{figure}

The GME condition is due to
\begin{equation}
\begin{aligned}
    \ev{\sigma_x} &= \tr[
        \pqty{\rho \otimes \bigotimes_{m=1}^{M-2}\varrho_m}
        \Pi_{+(2M-2)}
    ] \\
    &= \frac{\pi}{2}\widetilde{W}_{\tr_{-}\rho}\pqty{
        \alpha = 0;
        \mathcal{R} = \{\varrho_m\}_{m=1}^{M-2}
    },
\end{aligned}
\end{equation}
where the second line comes from the proof of \cref{thm:GMN-reduced-state}.
As such, this circuit implements a direct measurement of the smoothed Wigner function from \cref{thm:GMN-reduced-state}, which is used to certify the GME of the input state.

\subsection{\label{wit:BS-Rand}With controlled beam splitter and random displacements}
Similar to the witness in Sec.~\ref{wit:BS-A}, this witness directly measures $\widetilde{W}_{\tr_{-}\rho}(\alpha;\mathcal{R})$ from \cref{thm:GMN-reduced-state}.
However, instead of utilizing auxiliary modes, we shall choose $\mathcal{R}$ to be Gaussian states and simulate the convolution kernel using randomized measurement settings.

The implementation of this witness, with the full circuit illustrated in \cref{fig:wit:BS-Rand}, proceeds as follows:
\begin{enumerate}
    \item Choose $\alpha\in\mathbb{C}$ and $\Sigma\in\mathbb{R}^2 : \Sigma \succeq 0, \sqrt{\det\Sigma} \geq \frac{M-2}{4M^2}$
    \item For each run of the experiment,
    \begin{enumerate}[i.]
    \item Sample a random $\hat{\beta}$ such that
        \begin{equation}
            \pmqty{
                \Re[\hat{\beta}] \\
                \Im[\hat{\beta}]
            } \sim \mathcal{N}\pqty{
                -\frac{1}{\sqrt{M}}\pmqty{
                    \Re[\alpha] \\
                    \Im[\alpha]
                },
                \Sigma
            },
        \end{equation}
        where $\mathcal{N}(\vec{\mu},\Sigma)$ is a normal distribution with mean $\vec{\mu}$ and covariance matrix $\Sigma$
        \item Perform the \emph{same} displacement $D(\hat{\beta})$ on every mode
        \item Perform the controlled multiport beam splitter $\ketbra{\uparrow}\otimes\mathbbm{1} + \ketbra{\downarrow}\otimes \Pi_{+M}$ as defined in \cref{eq:VpmMDef}
        \item Measure $\sigma_x$ on the control qubit
    \end{enumerate}
    \item If the average over many experimental runs satisfies $\ev{\sigma_x} < 0$, then the input state is GME
\end{enumerate}

In the last step, the average of the qubit output over the random displacements is given by
\begin{equation}\label{eq:cbsrandomdisp}
\begin{aligned}
    \ev{\sigma_x} &= \int_{\mathbb{C}}\dd[2]{\beta} \frac{
        e^{-\frac{1}{2}\abs{\Sigma^{-\frac{1}{2}}\spmqty{
            \Re[\beta+\frac{\alpha}{\sqrt{M}}]\\
            \Im[\beta+\frac{\alpha}{\sqrt{M}}]
        }}^2}
    }{
        2\pi\sqrt{\det\Sigma}
    } \\
    &\qquad\qquad\qquad\qquad{}\times{}\tr[D(\beta\vec{1})\rho D^\dag(\beta\vec{1}) \Pi_{+M}] \\
    &= \int_{\mathbb{C}}\dd[2]{\beta} \frac{
        e^{-\frac{1}{2}\abs{\Sigma^{-\frac{1}{2}}\spmqty{
            \Re[\beta+\frac{\alpha}{\sqrt{M}}]\\
            \Im[\beta+\frac{\alpha}{\sqrt{M}}]
        }}^2}
    }{
        4\sqrt{\det\Sigma}
    } W_{\tr_{-}\rho}(-\sqrt{M}\beta) \\
    &= \int_{\mathbb{C}}\frac{\dd[2]{\beta}}{M} \frac{
        e^{-\frac{1}{2}\abs{(M\Sigma)^{-\frac{1}{2}}\spmqty{
            \Re[\beta-\alpha]\\
            \Im[\beta-\alpha]
        }}^2}
    }{
        4M^{-1}\sqrt{\det(M\Sigma)}
    } W_{\tr_{-}\rho}(\beta) \\
    &= \frac{\pi}{4(1-M^{-1})}\widetilde{W}_{\tr_{-}\rho}\pqty{\alpha;\mathcal{R}_G},
\end{aligned}
\end{equation}
where we performed a change of variable $-\sqrt{M}\beta \to \beta$ in the penultimate line and identified the kernel function as the filter $K(\alpha;\mathcal{R}_G)$ from \cref{eq:filter-function-Gaussian} with $\Sigma' = M\Sigma$.
As such, \cref{thm:GMN-reduced-state} implies that $\ev{\sigma_x} < 0$ certifies the GME of the initial state $\rho$.

\begin{figure}
    \centering
    \includegraphics{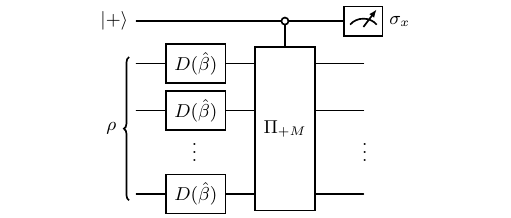}
    \caption{
        \label{fig:wit:BS-Rand}
        Witness in Sec.~\ref{wit:BS-Rand}: Implementation with a controlled beam splitter and random displacements.
        The output from the control qubit is averaged over many random displacements $\hat{\beta}$, which are sampled from a normal distribution.
    }
\end{figure}

In principle, this method can be extended to other distributions of $\hat{\beta}$ that correspond to other choices of $\mathcal{R}$ beyond Gaussian states.
Note however that $\mathcal{R}$ must be chosen such that the resulting filter function $K(\alpha;\mathcal{R})$, which relates to the required probability density of $\hat{\beta}$, is positive.
This is so that $\hat{\beta}$ can be sampled efficiently without requiring additional quantum resources.

\subsection{\label{wit:D}With controlled displacements}

\begin{figure}
    \centering
    \includegraphics{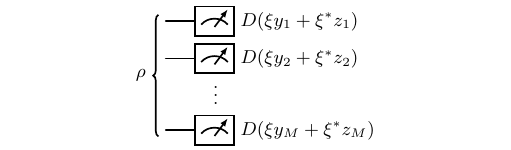}
    \caption{
        \label{fig:wit:D}
        Witness in Sec.~\ref{wit:D}: Implementation with local displacement measurements.
        The measurements are performed at a finite number of points $\xi \in \{\xi_n-\xi_{n'}\}_{n \leq n'}$.
    }
\end{figure}

Another implementation of \cref{thm:GMN-reduced-state} is based on direct measurements of the local characteristic function using the circuit illustrated in \cref{fig:wit:D}.
The witness is performed using the following steps.
\begin{enumerate}
    \item Choose coefficients $\vec{y},\vec{z} \in \mathbb{C}^M : \vec{y}\circ\vec{y}^* - \vec{z}\circ\vec{z}^* = \vec{1}$ and a set of phase-space points $\Xi = \{\xi_n\}_{n=1}^N$
    \item Measure $\langle D((\xi_n-\xi_{n'})\vec{y}+(\xi_n^*-\xi_{n'}^*)\vec{z})\rangle$ using the circuit in \cref{fig:wit:D} at $N(N+1)/2$ pairwise differences $\{\xi_n-\xi_{n'}:1 \leq n \leq n' \leq N\}$
    \item Collect the measurement outcomes into a finite matrix $\mathbf{C}(\rho;\Xi) \in \mathbb{C}^{N \times N}$ with entries $[\mathbf{C}(\rho;\Xi)]_{n,n'} = \frac{1}{N}\langle D((\xi_n-\xi_{n'})\vec{y}+(\xi_n^*-\xi_{n'}^*)\vec{z})\rangle$ for $n \leq n'$, while $\mathbf{C} = \mathbf{C}^\dagger$ is imposed to obtain the missing terms
    \item Choose $M-2$ auxiliary states $\mathcal{R} = \{\varrho_m\}_{m=1}^{M-2}$, and construct the kernel matrix $\mathbf{K}(\mathcal{R};\Xi) \in \mathbb{C}^{N\times N}$ with elements $[\mathbf{K}(\mathcal{R};\Xi)]_{n,n'}=\prod_{m=1}^{M-2} \chi_{\varrho_m}(\xi_{n}-\xi_{n'})$ by computing their characteristic functions
    \item Finally, if $\|\mathbf{C}(\rho;\Xi)\circ\mathbf{K}(\mathcal{R};\Xi)\|_1 > 1$, then the input state is GME
\end{enumerate}
Here, $\left\|\mathbf{A}\right\|_1$ denotes the trace norm, or sum of singular values, of $\mathbf{A}$.
Meanwhile, the GME condition from the last step comes from a direct application of \cref{prop:bochner-extended} on \cref{thm:GMN-reduced-state}.

Notice also that the last two steps are simply classical postprocessing procedures.
After constructing $\mathbf{C}(\rho;\Xi)$ from the measured data, the kernel $\mathbf{K}(\mathcal{R};\Xi)$ can be optimized numerically over the choice of auxiliary states $\mathcal{R} = \{\varrho_m\}_{m=1}^{M-2}$ at no extra experimental cost.

\section{States Detected by Corollary~\ref{col:GME-Wigner-witness}}
In this section, we expand upon the implementation of \cref{col:GME-Wigner-witness} through the witnesses in Secs.~\ref{wit:P}~and~\ref{wit:D-BS} by providing examples of detected GME states.
We first characterize different families of GME states that are directly revealed by the corollary with the integral performed over the whole of phase space, i.e., with $\omega = \mathbb{C}$, before going into the specifics of either witness.

Afterwards, we study how well these states are detected by the witness in Sec.~\ref{wit:P} when the integral is instead performed over a finite region of phase space $\omega \subsetneq \mathbb{C}$.
Last, we discuss how the same states can be detected with few measurement settings using the witness in Sec.~\ref{wit:D-BS}, which implements a hybrid Wigner and characteristic function measurement of the state.

\subsection{GME generated via vacuum interference}
A convenient family of states that exhibit GME are states prepared locally in one mode, then interfered with the vacuum using a balanced multiport beam splitter.
For such states, we can show the following property that is proven in Appendix~\ref{apd:center-of-mass-absolute-volume}.
\begin{proposition}\label{prop:center-of-mass-absolute-volume}
    Let $U_M$ be a multiport beam splitter such that $\forall m : \sqrt{M} U_M^\dag a_m U_M = a_1 + \sum_{m'\neq1}c_{m,m'} a_{m'}$.
    Given $\mathcal{U}_M(\rho_1) \coloneqq U_M(\rho_1\otimes\ketbra{0}^{\otimes {M-1}})U_M^\dag$, the absolute Wigner volume of $\mathcal{U}_M(\rho_1)$ along $\{\alpha\vec{1} : \alpha \in \mathbb{C} \}$ is
    \begin{equation}
        \mathcal{V}_{2D}(\mathcal{U}_M(\rho_1);\mathbb{C}) =  \frac{1}{M}\int_{\mathbb{C}}\dd[2]{\alpha} \abs{W_{\rho_1}(\alpha)}.
    \end{equation}
\end{proposition}
This gives a sufficient condition for states of the form $\mathcal{U}_M(\rho_1)$ to be detected by \cref{col:GME-Wigner-witness}.
Since it is known that the absolute Wigner volume of the Fock states $\ket{n}$ grow as $\sim \sqrt{n}$ \cite{AnatoleNegativity2004}, $\mathcal{V}_{2D}(\mathcal{U}_M(\ket{n});\mathbb{C}) \sim \sqrt{n}/M$ implies that the GME of every state $\mathcal{U}_M(\ket{n})$ with $n \gtrsim M$ can be detected by this corollary.

A notable member of this family is the $W$ state $\mathcal{U}_M(\ket{1}) = \ketW \propto \ket{10\cdots00} + \ket{01\cdots00} + \dots \ket{00\cdots01}$ encoded in the Fock basis, for which
\begin{equation}
    \mathcal{V}_{2D}(\ketW;\mathbb{C})
    = \frac{4}{M\sqrt{e}} - \frac{1}{M},
\end{equation}
and as such the GME of $\ket{W_\text{M}}$ is detected when
\begin{equation}
\begin{gathered}
\begin{aligned}
    \frac{4}{M\sqrt{e}} - \frac{1}{M} &> \frac{1}{2\sqrt{M-1}} \\
    \pqty{\frac{M}{2}-\pqty{\frac{4}{\sqrt{e}}-1}^2}^2 &> \pqty{\frac{4}{\sqrt{e}} - 1}^4 - \pqty{\frac{4}{\sqrt{e}} - 1}^2
\end{aligned} \\
\implies 1.167 < M < 6.968.
\end{gathered}
\end{equation}
Hence, the GME of all $W$ states for $M \leq 6$ modes can be detected by \cref{col:GME-Wigner-witness}. In \cref{fig:thm1-W-pure}, we plot the corresponding two-dimensional slices of the $W$ states, and the exact amount of violation $\mathcal{V}_{2D}-(2\sqrt{M-1})^{-1}$ achieved when implementing the witness for $M$ modes.

\begin{figure}
    \begin{center}
	\includegraphics[width=\columnwidth]{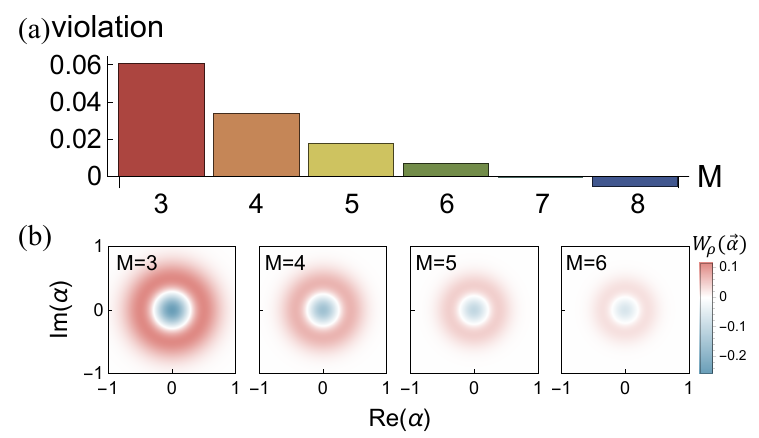}
	\end{center}
    \caption{
        \label{fig:thm1-W-pure}
        Detecting the GME of pure $M$-mode $W$ states $\ketW$ using \cref{col:GME-Wigner-witness}.
        (a) The plot of the violation $\mathcal{V}_{2D}(\ketW,\mathbb{C}) - (2\sqrt{M-1})^{-1}$ against $M$ for the two-dimensional slice $\vec{\alpha}=\alpha\vec{1}$.
        GME is detected when $\mathcal{V}_{2D} > (2\sqrt{M-1})^{-1}$.
        (b) The Wigner function $W_{\ketW}(\vec{\alpha})$ along $\vec{\alpha} = \alpha\vec{1}$ for different number of modes $M$.
    }
\end{figure}

Another member of this family of states are entangled cat states, which are $M$-mode Greenberger--Horne--Zeilinger states encoded in the coherent basis as
\begin{equation}
    \ketCat = \frac{\ket{\gamma}^{\otimes M} +
        \ket{-\gamma}^{\otimes M}}{\sqrt{2(1+e^{-2M\abs{\gamma}^2})}}
    = \mathcal{U}_M(\iketCat[1][\sqrt{M}\gamma]),
\end{equation}
where $|\gamma|$ denotes the cat size.
The Wigner function of this state reads
\begin{equation}
\begin{aligned}
    W_{\ketCat}(\vec{\alpha}) &=
    \pqty{\frac{2}{\pi}}^M
    \frac{
        e^{-2\abs{\vec{\alpha}}^2}
    }{
        1+e^{-2M\abs{\gamma}^2}
    }\Bigg[\cos(4\Im[\vec{1}^T\vec{\alpha}\gamma^*])\\
    &\qquad{}+{}e^{-2M\abs{\gamma}^2}
    \cosh(4\Re[\vec{1}^T\vec{\alpha}\gamma^*])
    \Bigg].
\end{aligned}
\end{equation}
With this, we can compute $\mathcal{V}_{2D}(\iketCat;\mathbb{C})$ along the $\vec{\alpha} = \alpha\vec{1}$ slice.
This is shown in \cref{fig:Catsize}, where we present the amount of violation $\mathcal{V}_{2D}-(2\sqrt{M-1})^{-1}$ obtained when applying the GME criterion in \cref{col:GME-Wigner-witness} to entangled $M$-mode cat states $\ketCat$.

\begin{figure}
    \begin{center}
	\includegraphics[width=\columnwidth]{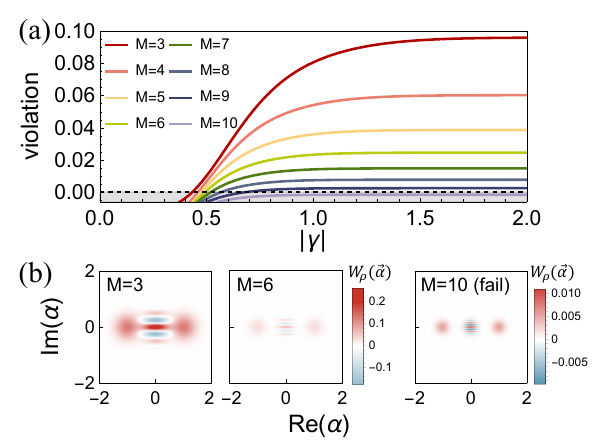}
	\end{center}
    \caption{
        \label{fig:Catsize}
        Detecting the GME of pure $M$-mode entangled cat states $\ketCat$ with \cref{col:GME-Wigner-witness}.
        (a) Amount of violation $\mathcal{V}_{2D}(\ketCat,\mathbb{C}) - (2\sqrt{M-1})^{-1}$ as a function of cat size $|\gamma|$.
        The maximum number of modes that GME can be detected is $M=9$.
        (b) The joint Wigner function $W_{\ketCat}(\vec{\alpha})$  on the two-dimensional slice for cat states of size $\gamma=1$, for $M \in \{3,6,10\}$.
    }
\end{figure}

\subsection{\label{sec:beyond-vacuum-1}GME beyond vacuum interference}
While states of the form $\mathcal{U}_M(\rho_1)$ are a convenient family to study, the generation of GME through vacuum interference is by no means necessary for GME to be detected by our witnesses.
In this section, we give some examples of the detection of GME states that do not arise from vacuum interference.

Our first example is the two-excitation Dicke state \cite{dicke-states}
\begin{equation}
\begin{aligned}
    \iket{D_2^M} &\coloneqq \sqrt{\frac{2/M}{M-1}} \sum_{m=1}^{M-1} \sum_{m'=m+1}^{M} a_m^\dag a_{m'}^\dag \ket{0}^{\otimes M} \\
    &= \sqrt{\frac{2/M}{M-1}}\bigg(
        \ket{110\cdots00} + \ket{101\cdots00} \\
        &\qquad\qquad\qquad {}+{} \cdots + \ket{000\cdots11}
    \bigg),
\end{aligned}
\end{equation}
which is a superposition of all permutations of $\ket{1}^{\otimes 2} \otimes \ket{0}^{\otimes (M-2)}$.
Its Wigner function along $\alpha\vec{1}$ is
\begin{equation}
\begin{aligned}
    &W_{\iket{D_2^{M}}}(\alpha\vec{1}) \\
    &\quad{}={}\frac{2^{M} e^{-2M\abs{\alpha}^2}}{\pi^M}\bqty{
        1 + 8 (M-1)\abs{\alpha}^2(M\abs{\alpha}^2-1)
    },
\end{aligned}
\end{equation}
from which its absolute volume is found to be
\begin{equation}
\begin{aligned}
    &\mathcal{V}_{2D}(\iket{D^M_2};\mathbb{C}) \\
    &\quad{}={}
    \pqty{\frac{\pi}{2}}^{M-1} \int_{\mathbb{C}}\dd[2]{\alpha} \abs{W_{\iket{D_2^{M}}}(\alpha\vec{1})} \\
    &\quad{}={} \frac{1}{M} - \frac{16(M-1)}{e M^2}\sinh\sqrt{\frac{M-2}{2(M-1)}} \\
    &\qquad\qquad{}+{} \frac{8\sqrt{2(M-2)(M-1)}}{e M^2}\cosh\sqrt{\frac{M-2}{2(M-1)}}.
\end{aligned}
\end{equation}

\begin{figure}
    \begin{center}
	\includegraphics[width=\columnwidth]{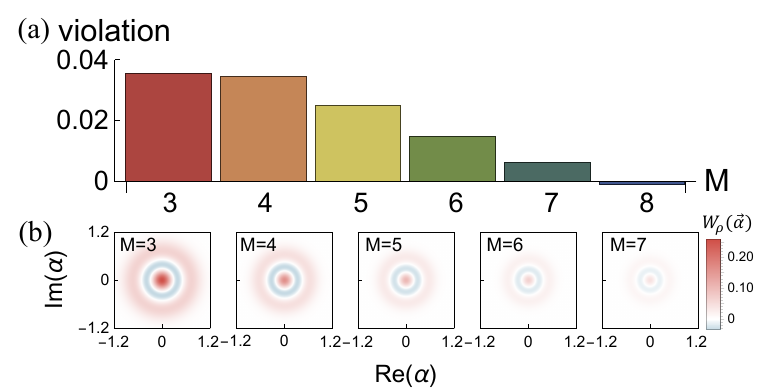}
	\end{center}
    \caption{
        \label{fig:Dicke}
        Detecting the GME of two-excitation Dicke states $\iket{D_2^M}$ using \cref{col:GME-Wigner-witness}.
        (a) The violation $\mathcal{V}_{2D}(\iket{D_2^M};\mathbb{C}) - (2\sqrt{M-1})^{-1}$ against $M$ for the two-dimensional slice $\vec{\alpha}=\alpha\vec{1}$.
        GME is detected when $\mathcal{V}_{2D} > (2\sqrt{M-1})^{-1}$.
        (b) The Wigner function $W_{\iket{D_2^M}}(\vec{\alpha})$ along $\alpha = \alpha\vec{1}$ for $3 \leq M \leq 7$.
    }
\end{figure}

The amount of violation $\mathcal{V}_{2D}-(2\sqrt{M-1})^{-1}$ is plotted in \cref{fig:Dicke}, where we find that $\mathcal{V}_{2D}>(2\sqrt{M-1})^{-1}$ for all $M \leq 7$.
Therefore, the GME of the two-excitation Dicke state is detected by \cref{col:GME-Wigner-witness} for $3 \leq M \leq 7$.

The W and Dicke states are permutation symmetric and eigenstates of the total energy, which makes the analytical forms of their absolute Wigner volume easier to work with.
However, this is purely for theoretical convenience, and not necessary for the witness.
Take the following, more complex, tripartite states as examples
\begin{equation}
\begin{aligned}
    \ket{\psi_1}&\coloneqq
        \frac{3+\sqrt{23}}{8\sqrt{2}}\ketW[3] +
        \frac{\sqrt{3}}{4}\ket{D_2^3} +
        \frac{\sqrt{23}-3}{8}\ket{011} \\
    &\qquad{}+{} \frac{\sqrt{23} - 1}{8\sqrt{2}}\ket{200} +
        \frac{1}{4\sqrt{2}}\ket{020} +
        \frac{1}{4\sqrt{2}}\ket{002} \\
    \ket{\psi_2} &\coloneqq \frac{1}{3\sqrt{2}}\ket{300} + \frac{1}{\sqrt{2}}\ket{D_2^3} -
    \frac{1}{2\sqrt{6}}\ket{1} \otimes \pqty{\ket{20}+\ket{02}}\\
    &\qquad{}+{} \frac{1}{2\sqrt{3}}\ket{200} + \frac{1}{2\sqrt{3}}\ket{020} + \frac{1}{2\sqrt{3}}\ket{002} \\
    &\qquad{}+{} \frac{1}{6\sqrt{2}}\ket{0}\otimes\pqty{
        \sqrt{3}\ket{12} + \sqrt{3}\ket{21} -
        \ket{30} - \ket{03}
    }.
\end{aligned}
\end{equation}
Here $\ket{\psi_1}$ and $\ket{\psi_2}$ are not symmetric with respect to any permutation involving the first mode and contain superpositions of states with one and two total excitations for $\ket{\psi_1}$ and two and three total excitations for $\ket{\psi_2}$.
Their Wigner functions along $\alpha\vec{1}$ take the form
\begin{equation}
\begin{aligned}
    W_{\ket{\psi_1}}(\alpha\vec{1}) &= \frac{e^{-6\abs{\alpha}^2}}{16\pi^3}\bigg(
        (16 + 3\sqrt{23})\pqty{
            12\abs{\alpha}^2 - 1
        }^2 \\
    &\qquad{}+{}
        8\sqrt{6}\Re[\alpha](16 + 3\sqrt{23})\pqty{
            6\abs{\alpha}^2 - 1
        }
        \\[-1ex]
    &\qquad{}+{} 48 - 15\sqrt{23}
    \bigg) \\
    W_{\ket{\psi_2}}(\alpha\vec{1}) &= \frac{4e^{-6\abs{\alpha}^2}}{\pi^3}
    \pqty{
        1 - 30 \abs{\alpha}^2 + 108 \abs{\alpha}^4
    }.
\end{aligned}
\end{equation}
Directly integrating these expressions, we obtain $\mathcal{V}_{2D}(\ket{\psi_1};\mathbb{C}) \gtrsim 0.445$ and $\mathcal{V}_{2D}(\ket{\psi_2};\mathbb{C}) \gtrsim 0.424$, which are both larger than $(2\sqrt{3-1})^{-1} \lesssim 0.353$.
Therefore, the GME of both states can be detected using \cref{col:GME-Wigner-witness}.

\subsection{Robustness in the presence of noise}
To study the robustness of the witnesses in experimental settings, we need to consider how noise might affect the detectability of different states.
For many qubit-CV systems, a common source of noise is the loss of energy to the environment.
This is modelled by the amplitude damping channel $\mathcal{D}_\eta(\rho)$ that mixes every system mode $a_m \to \sqrt{1-\eta}a_m + \sqrt{\eta}a^{(E)}_m$ with a corresponding environment mode $a^{(E)}_m$, where the modes $\{a_m^{(E)}\}_{m=1}^M$ are later traced out.

When applied to the $M$-mode $W$ state, we have
\begin{equation}
\rhoW \coloneqq
\mathcal{D}_\eta(\ketW)
= (1-\eta)\ketbraW + \eta \ketbra{0}^{\otimes M}  .
\end{equation}
To evaluate the GME criterion from \cref{col:GME-Wigner-witness}, we examine the joint Wigner function of the noisy $W$ state on the two-dimensional slice $\vec{\alpha} = \alpha\vec{1}$.
This is given by
\begin{equation}
    W_{\rhoW} (\alpha\vec{1}) = \pqty{\tfrac{2}{\pi}}^M e^{-2 M|\alpha|^2}\left[(1-\eta)4 M|\alpha|^2+2\eta-1\right],
\end{equation}
from which the violation $\mathcal{V}_{2D}-(2\sqrt{M-1})^{-1}$ of the biseparability bound can be directly computed.
This violation is plotted against the loss parameter $\eta$ in \cref{fig:thm1-W-noisy}.
We find that the witnesses are robust against loss for small number of modes but becomes less robust as $M$ increases.
The threshold amount of noise, above which GME cannot be detected using this criterion, is $\eta_{\text{max}}=\{ 0.3548, 0.2446, 0.1522, 0.0710 \}$ for $M=\{3,4,5,6\}$, respectively, which shrinks approximately as $\eta_{\max} < 0.65 - 0.095M$.

\begin{figure}
    \begin{center}
	\includegraphics[width=\columnwidth]{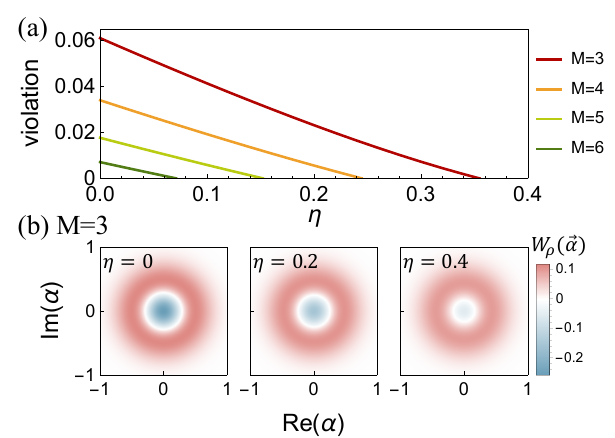}
	\end{center}
    \caption{\label{fig:thm1-W-noisy}
        Detecting the GME of $W$ states $\rhoW$ under damping loss with \cref{col:GME-Wigner-witness}.
        (a) Amount of violation $\mathcal{V}_{2D}(\rhoW;\mathbb{C}) - (2\sqrt{M-1})^{-1}$ against the loss parameter $\eta$, and (b) the joint Wigner function $W_{\rhoW[3]}(\vec{\alpha})$ along the two-dimensional slice $\vec{\alpha} = (\alpha,\alpha,\alpha)$ at different loss parameters $\eta$.
    }
\end{figure}

As another example, we apply the loss channel to the $M$-mode cat state to obtain $\mathcal{D}_\eta(\ketCat) = \rhoCat$.
The Wigner function of the resulting lossy cat state along the two-dimensional slice $\alpha\vec{1} = (\alpha,\alpha,\dots,\alpha)$ is
\begin{equation}
\begin{aligned}
&W_{\rhoCat}(\alpha\vec{1} ) = \frac{e^{-2M\abs{\alpha}^2}}{1+e^{-2M\abs{\gamma}^2}}
\pqty{\frac{2}{\pi}}^M\Bigg[ \\
&\qquad e^{-2M\abs{\gamma}^2(1-\eta)} \cosh(
    4M \sqrt{1-\eta} \Re[\alpha\gamma^*]
)\\
&\quad\qquad{}+{}  e^{-2M\abs{\gamma}^2\eta} \cos( 4M\sqrt{1-\eta} \Im[\alpha \gamma^*]
)  \Bigg].
\end{aligned}
\end{equation}
Figure~\ref{fig:CatM3GMELoss} shows the GME detection threshold for $M$-mode cat states subject to the amplitude-damping channel $\mathcal{D}_\eta$, where the maximum tolerable loss parameter $\eta_{\max}$ is plotted as a function of the cat size $|\gamma|$. It is observed that $\eta_{\text{max}}$ increases with $|\gamma|$ within the region $0 < |\gamma|< \gamma_{M,\max}$, where $\gamma_{M,\max} \in \{0.828,0.745,0.716,0.721,0.709,0.760,0.863\}$ for $M \in \{3,4,5,6,7,8,9\}$.
However, $\eta_{\text{max}}$ decreases again as the cat size grows above $|\gamma| > \gamma_{M,\max}$.
This is expected behavior, as cat states are known to be more susceptible to the effects of energy loss for larger cat sizes \cite{cat-state-decoherence}.

\begin{figure}
    \begin{center}
	\includegraphics[width=\columnwidth]{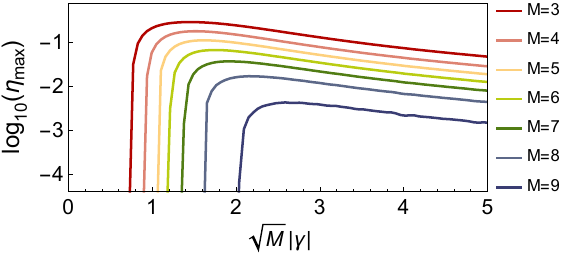}
	\end{center}
    \caption{
        The threshold loss parameter $\eta_{\max}$, plotted in log base 10, above which the witness in Sec.~\ref{wit:P} fails to detect the GME of the $M$-mode cat state, plotted against the cat size $\sqrt{M}|\gamma|$.
    }
    \label{fig:CatM3GMELoss}
\end{figure}

\subsection{\label{sec:P-finite}With numerical integration of Wigner function over finite regions}

\begin{figure}
    \begin{center}
	\includegraphics[width=\columnwidth]{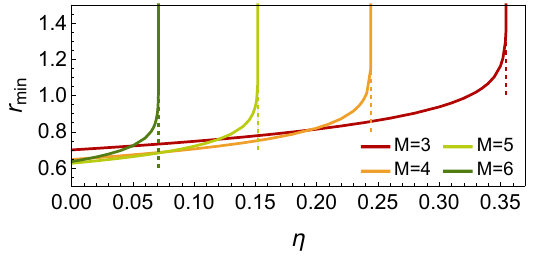}
	\end{center}
    \caption{
        \label{fig:Wrmin}
        Minimum radius $r_{\min}$ required to detect lossy $W$ states $\rho_{\text{W},\eta}$ using the witness in Sec.~\ref{wit:P}, where the absolute integral is taken over a circle $0 \leq \abs{\alpha} \leq r$.
        We find that $r_{\min}$ increases with loss parameter $\eta$, and diverges at the maximum loss parameter $\eta_{\max}$ (dashed lines), as found in the previous section.
    }
\end{figure}

In this section, we investigate the effects of two pragmatic simplifications made when implementing \cref{col:GME-Wigner-witness}.
First, the required integral would only be performed over a finite region in actual experiments, so we will quantify how large this region needs to be to detect the previous examples of GME states.
Second, the required integral will be computed numerically in practice, so error estimates of the results of numerical integration will also be analyzed.

We begin with the lossy $W$ state $\rhoW$.
Since its joint Wigner function is centrally symmetric along the two-dimensional slice $\vec{\alpha} = \alpha\vec{1}$, as seen in \cref{fig:thm1-W-noisy}, we shall consider the implementation of \cref{col:GME-Wigner-witness} using the absolute integral of $W_{\rhoW}(\alpha\vec{1})$ over the finite circular region $0 \leq \abs{\alpha} \leq r$.
We find that, for a given $M$ and $\eta$, there is a minimum radius $r_{\min}$ below which the integral no longer violates the GME bound.
This minimum radius is plotted against the amount of loss in \cref{fig:Wrmin} for different number of modes $M$, where we find that $r_{\min}$ increases with $\eta$.
We also find that $r_{\min}$ diverges at a certain loss parameter, which corresponds exactly to the maximum loss $\eta_{\max}$ that a lossy $W$ state can be detected with an integral over the entire two-dimensional slice, as found in the previous section.

We now turn to the second practical issue, which is that the required integral will be computed using numerical integration (i.e., discrete summation) in actual experimental implementations.
Thus, reasonable error estimates are needed in order to minimize the possibility of false positives due to numerical errors.
Below, two error analysis techniques---one rigorous and one heuristic---are presented as basic examples, but we refer to any numerical integration textbook for deeper discussions \cite{numerical-integration}.

Let us first present a rigorous error analysis, which requires an additional piece of information: an upper bound $E$ on the total energy $\sum_m\langle a^\dag_m a_m \rangle \leq E$ of the system.
This can be obtained by additionally measuring the energy of the system when implementing the witness.
Alternatively, an implicit bound might be present due to some property of the physical system---for example, a trapped object must certainly be less energetic than its trap, while qubit-CV systems working in the dispersive regime must satisfy certain energy conditions.

Either way, an energy bound gives rise to a rigorous bound for the integral discretized as a square grid, laid out in the following proposition.

\begin{proposition}\label{prop:rigorous-integral-error}
    Let $\{(n_R+in_I)\Delta \}_{(n_R,n_I) \in \mathcal{S} \subseteq \mathbb{Z}^2}$ be a square grid in $\mathbb{C}$ with grid spacing $\Delta$.
    Given that an upper bound $E$ on the total energy of the system $\sum_{m=1}^M{\tr}[\rho \, a^\dagger_m a_m]$ is known, define
    \begin{equation}
    \begin{aligned}
        &\phantom{\delta} I(\rho,\Delta;\mathcal{S}) \coloneqq \hspace{-1em}\sum_{(n_R,n_I)\in\mathcal{S}} \hspace{-1em} \Delta^2\abs\Big{
            \tr[\rho\, \Pi\pqty{
                (n_R+in_I)\Delta \vec{1}
            }]
        } \\
        &\delta I(\Delta;\mathcal{S},E) \\[-1ex]
        &\qquad{}\coloneqq{}\hspace{-1em}\sum_{(n_R,n_I)\in\mathcal{S}}\hspace{-1em}
            2\Delta^3\sqrt{2ME} +
            2M\Delta^4\pqty{1 + \sqrt{2(n_R^2+n_I^2)}}.
    \end{aligned}
    \end{equation}
    Then, $I(\rho,\Delta;\mathcal{S}) - \delta I(\Delta;\mathcal{S},E) > \pi(4\sqrt{M-1})^{-1}$ implies that $\rho$ is GME.
\end{proposition}
The proof is given in Appendix~\ref{apd:rigorous-integral-error}.
Here $I$ is to be understood as the computed value of the numerical integration, while $\delta I$ is a rigorous bound on its error.

Let us study this technique with the three-mode $W$ state $\ketW[3]$, which has the energy bound $E = \sum_{m=1}^3\braW[3]a_m^\dag a_m\ketW[3] = 1$.
The numerical integral will be performed over the region $\mathcal{S}$ that approximates a circle centered at the origin with radius $r$, as illustrated in Figs.~\ref{fig:numerical-integration}(a) and \ref{fig:numerical-integration}(b) for grid spacings $\Delta \in \{0.2,0.1\}$.
The computed integral $I(\ketW[3],\Delta)-\delta I(\Delta)$ is plotted against $r$ in \cref{fig:numerical-integration}(c) for different choices of grid spacing $\Delta$.
We find that a grid spacing of at least $\Delta < 0.0058$ with radius $r \approx 0.9$ is required for GME of $\ketW[3]$ to be detected within these rigorous bounds.

\begin{figure}
    \centering
    \includegraphics[width=\columnwidth]{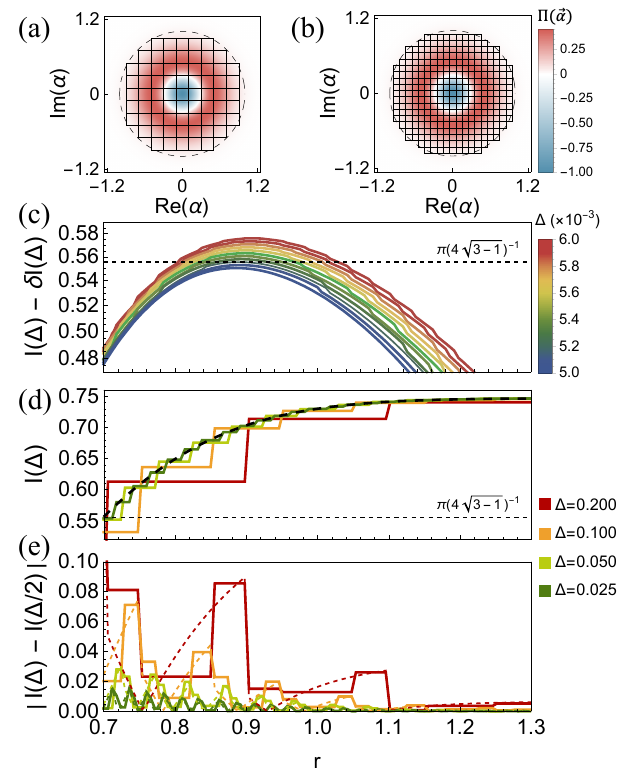}
    \caption{
        \label{fig:numerical-integration}
        We consider the three-mode $W$ state over the two-dimensional slice $\alpha_m=\alpha$ and its integral over the circular region $\abs{\alpha} < r$.
        This will be numerically computed by discretizing and approximating the integration region with a square grid $\mathcal{S}$, shown here for $r=1$ and grid spacings (a) $\Delta=0.2$ and (b) $\Delta=0.1$.
        (c) The numerical integral $I(\ketW[3],\Delta;\mathcal{S})$ minus the rigorous error bound $\delta I(\Delta;\mathcal{S},E=1)$ against the radius $r$.
        GME is detected when $I-\delta I > \pi(4\sqrt{3-1})^{-1}$.
        We find that the grid spacing must satisfy $\Delta < 0.0058$ for a theoretical guarantee that GME was detected.
        (d) The numerical integral  $I(\ketW[3],\Delta;\mathcal{S})$ for larger gap sizes $10^{-3} \lesssim \Delta \lesssim 10^{-1}$.
        While they do not violate the GME bound when the rigorous errors are accounted for, we find that computed quantities are still very close to the actual integral plotted in black.
        (e) Heuristically, the error can be approximated by taking the difference $|I(\ketW[3],\Delta;\mathcal{S})-I(\ketW[3],\Delta/2;\mathcal{S})|$ of two numerical integrals performed with different grid spacings (solid lines).
        While heuristic, it is in agreement to the same order of magnitude to the actual errors $|\mathcal{V}_{2D}(\ketW[3])-I(\ketW[3],\Delta;\mathcal{S})|$ (dashed lines).
    }
\end{figure}

This rigorous error estimate should be used if, under the assumptions that the measurements and energy bounds are well characterized, theoretical guarantees about the presence of GME are required.
Unfortunately, it is still experimentally demanding, as the required grid spacing and area of integration requires $\pi r^2/\Delta^2 \approx 76\,000$ phase-space points to be measured.

For a more tractable approach, we turn to heuristic error approximation methods.
First, we compute the numerical integral $I(\rho,\Delta)$ using different grid spacings, as shown in \cref{fig:numerical-integration}(d).
Then, we calculate the difference between the computed integral for two grid spacings $\Delta$ and $\Delta/2$, and take it to be approximately the error $|I(\rho,\Delta) - I(\rho,\Delta/2)| \approx |\mathcal{V}_{2D}(\rho)-I(\rho,\Delta)|$.
Both the actual and approximate errors are plotted in \cref{fig:numerical-integration}(e), where we find that the approximate errors are excellent order-of-magnitude estimates for the actual integration errors.
Furthermore, for all chosen grid spacings, both errors are suppressed to $\lesssim 0.02$ when the integral exceeds the radius $r > 0.9$.

Computing the same integral twice with a more and less precise algorithm, then taking their difference as an error estimate, is a common practice in numerical integration \cite{numerical-integration-error}.
Therefore, although no theoretical guarantee can be provided using this heuristic technique, it would at least strongly suggest the presence of GME in the system that might be detected using stronger tests.

We conclude this section by noting that the field of numerical integration is already very well-established: Sophisticated integration strategies, error estimation techniques, and even software libraries that perform automatic integration are widely available in the literature.
These techniques go beyond the scope of this work, and thus we shall refer the reader to any textbook dedicated to numerical integration for further details \cite{numerical-integration}.

\subsection{With beam splitters and measurements of finite characteristic function points}

In place of integration over a finite region, the witness in Sec.~\ref{wit:D-BS} detects GME with only a finite number of measurement settings.
It involves first choosing some coefficients $\vec{y},\vec{z} \in \mathbb{C}^M : \vec{y}\circ\vec{y}^* - \vec{z}\circ\vec{z}^* = \vec{1}$ and $N$ phase-space points $\Xi = \{\xi_n\}_{n=1}^N$ and then collecting the measurement outcomes of the system after performing a controlled beam splitter into the matrix $\mathbf{C} \in \mathbb{C}^{N \times N}$ as
\begin{equation}
\begin{aligned}
    [\mathbf{C}(\rho;\Xi)]_{n,n'} &= \frac{1}{N}\ev{\sigma_x \otimes  D(\xi_n-\xi_{n'})} \\
    &= \tr[\rho \; D_{+}\pqty{\sqrt{M}(\xi_n-\xi_{n'})} \; \Pi_{-M}],
\end{aligned}
\end{equation}
where $D_{+}$ is the displacement operator on the center-of-mass mode $\sqrt{M} a_+ = \sum_{m=1}^M a_m$, and $\Pi_{-M}$ is the parity operator on the relative modes $\{a_{-m}\}_{m=2}^M$ that satisfy the canonical commutation relations with $a_+$.

For a given $\Xi$, at most $N(N+1)/2$ measurement settings are required, with the other matrix elements imposed by the Hermiticity of $\mathbf{C}^\dagger = \mathbf{C}$.
In order to study the actual number of points needed, let us consider the implementation of this witness on some example states.

For the lossy $W$ state, the matrix elements of $\mathbf{C}(\rhoW;\Xi)$ are
\begin{equation}
\begin{aligned}
    [\mathbf{C}(\rhoW)]_{n,n'} = e^{-\frac{
        M\abs{\xi_n-\xi_{n'}}^2
    }{
        2
    }}\pqty{
        1 - M(1-\eta)\abs{\xi_n-\xi_{n'}}^2
    }.
\end{aligned}
\end{equation}
With this expression, we can explicitly construct the matrix and compute its trace norm $\|\mathbf{C}(\rhoW;\Xi)\|_1$, which we plot against $\eta$ for different dimensions $N$ in \cref{fig:witDBS_W}.

\begin{figure}
    \centering
    \includegraphics[width=\columnwidth]{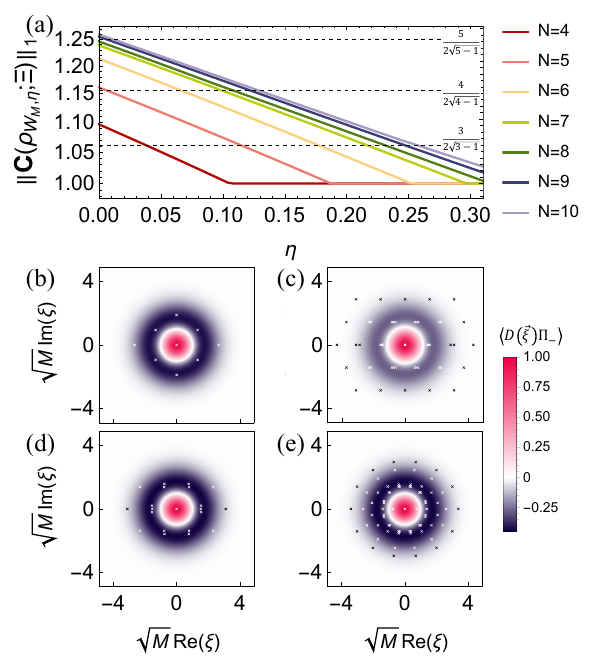}
    \caption{
        \label{fig:witDBS_W}(a) Implementation of the witness in Sec.~\ref{wit:D-BS} for noisy $W$ states, which only requires a finite number of measurement settings.
        The dashed horizontal lines are the GME thresholds: GME is detected when $\|\mathbf{C}\|_1 > M/(2\sqrt{M-1})$.
        Note that $\|\mathbf{C}\|_1$ is defined the same way for all $M$.
        With the property $\mathbf{C}=\mathbf{C}^\dag$, the number of measurement settings scales as $\sim N(N+1)/2$ where $N$ is the dimension of $\mathbf{C}$.
        Usually, fewer settings are needed due to duplicate entries in the matrix.
        For example, we find that (b) $5$ settings are enough to detect the $M=3$ $W$ state up to $\eta \leq 0.035$, (c) $27$ for $M=3$ up to $\eta \leq 0.255$, (d) $11$ settings for $M=4$ up to $\eta \leq 0.004$, and (e) $37$ settings for $M=5$ up to $\eta \leq 0.004$.
    }
\end{figure}
To obtain the plotted data, we numerically optimized $\|\mathbf{C}(\rhoW;\Xi)\|_1$ over $\Xi$ many times, each time starting with $N$ random points in phase space for $\Xi$, then reporting the maximum $\left\|\mathbf{C}\right\|_1$ found over all starting points.
The required measurement settings $\{\xi -\xi' : \xi,\xi' \in \Xi\}$ that correspond to the optimal $\left\|\mathbf{C}\right\|_1$ are reported for several choices of $\eta$ and $N$ in \cref{fig:witDBS_W}(b-e).
Generally, we find that the number of measurement settings required increases with both the loss parameter $\eta$ and the number of modes $M$.

That said, one still requires only $\sim 45$ phase-space points to detect the GME of a $W$ state for $3 \leq M \leq 5$.
This requires fewer points than the numerical integral in \cref{sec:P-finite} for $M=3$ with a large grid spacing $\Delta = 0.2$ and radius $r=1$, which requires measurements of the Wigner function at $\sim \pi r^2 /\Delta^2 \approx 78$ phase-space points.
Furthermore, unlike in the case of \cref{sec:P-finite}, we can theoretically guarantee the presence of GME with $\sim 45$ phase-space points for this witness, provided that the implemented measurements are well-characterized.

Performing the same analysis for the lossy entangled cat state, we similarly find that the matrix elements of $\mathbf{C}(\rhoCat;\Xi)$ are
\begin{equation}
\begin{aligned}
    &[\mathbf{C}(\rhoCat)]_{n,n'} = \frac{e^{-\frac{M\abs{\xi_n-\xi_{n'}}^2}{2}}}{1+e^{-2M\abs{\gamma}^2}}\bigg(\\
        &\quad\qquad\cos(2M\sqrt{1-\eta}\Im[(\xi_n-\xi_{n'})\gamma^*]) \\
        &\quad\qquad{}+{}
        e^{-2M\abs{\gamma}^2}\cosh(2M\sqrt{1-\eta}\Re[(\xi_n-\xi_{n'})\gamma^*])
    \bigg).
\end{aligned}
\end{equation}
As before, we explicitly construct $\mathbf{C}(\rhoCat;\Xi)$ with numerically optimized phase-space points $\Xi$, and plot its trace norm $\|\mathbf{C}(\rhoCat;\Xi)\|_1$ against the loss parameter $\eta$ and cat size $|\gamma|$ for different dimensions $N$ in \cref{fig:cat_state_hybrid}.

\begin{figure}
    \centering
    \includegraphics[width=\columnwidth]{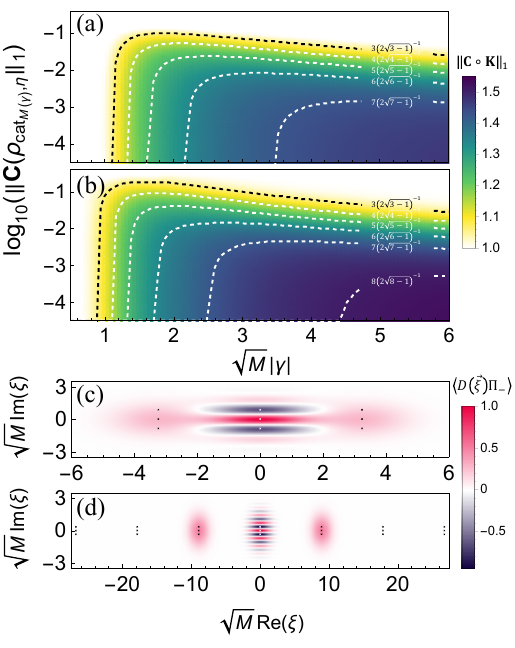}
    \caption{
        \label{fig:cat_state_hybrid}Implementation of the witness in Sec.~\ref{wit:D-BS} for noisy entangled cat states with magnitude $\sqrt{M}\abs{\gamma}$ for (a) $N=4$ and (b) $N=8$.
        The vertical axis is plotted with a base 10 log scale, dashed lines are $M(2\sqrt{M-1})^{-1}$, and GME is detected when $\|\mathbf{C}\|_1$ exceeds this value.
        Only a finite number of settings are needed to detect GME: (c) 5 settings for $M=3$ and $\eta = 0.1$, and (d) 11 settings for $M=8$ and $\eta = 5 \times 10^{-4}$.
    }
\end{figure}

As expected, the violation $\left\|\mathbf{C}\right\|_1-M(2\sqrt{M-1})^{-1}$ decreases with an increase in noise.
The GME of entangled cats for $M \leq 7$ modes can be detected with just $N=4$ dimensions of $\mathbf{C}$, while least $N=8$ is required for $M=8$ modes, although the actual number of measurement settings needed are fewer than for the $W$ state.
We also find that there is an initial increase in $\left\|\mathbf{C}\right\|_1$ against the cat size, but it slowly decreases after a threshold size is met.
This occurs because cat states are more susceptible to energy loss as they become larger \cite{cat-state-decoherence}.

\section{States Detected by Theorem~\ref{thm:GMN-reduced-state}}
We now turn to examples of states detected by Theorem~\ref{thm:GMN-reduced-state} and the witnesses in Secs.~\ref{wit:BS-A}~to~\ref{wit:D}. Like before, we shall first characterize the states detected directly by the theorem, before turning to some specific properties about its associated witnesses. In this case, the witness in Sec.~\ref{wit:BS-A} is a direct implementation of the theorem and therefore will not require a dedicated subsection to study its properties; while any state detected by the theorem is guaranteed to be detected by the witness in Sec.~\ref{wit:D}, as long as enough points of the characteristic function are measured, due to Bochner's theorem \cite{Bochner-theorem-1,Bochner-theorem-2}.

\subsection{GME generated via vacuum interference}
Like with the previous theorem, it is again convenient to consider the family of GME states prepared by interfering a state with the vacuum.
\begin{proposition}\label{prop:center-of-mass-nonclassicality-depth}
    Let $U_M$ be a multiport beam splitter such that $\forall m : \sqrt{M} U_M^\dag a_m U_M = a_1 + \sum_{m'\neq1}c_{m,m'} a_{m'}$.
    Given $\mathcal{U}_M(\rho_1) \coloneqq U_M(\rho_1\otimes\ketbra{0}^{\otimes {M-1}})U_M^\dag$,
    we have $\tau_c(\tr_{-}[\mathcal{U}_M(\rho_1)]) = \tau_c(\rho_1)$.
\end{proposition}
The proof is given in Appendix~\ref{apd:center-of-mass-nonclassicality-depth}.
The consequence of this proposition, together with \cref{col:nonclassicality-depth-reduced-state}, means that $\tau_c(\rho_1) > 1 - M^{-1}$ is sufficient for $\mathcal{U}_M(\rho_1)$ to be detected by \cref{thm:GMN-reduced-state}.

This family of detected states contains familiar examples: It is known that $\tau_c(\ket{1}) = \tau_c(\iketCat[1][\sqrt{M}\gamma]) = 1 > 1-M^{-1}$ \cite{nonclassicality-depth}, thus both the number-basis $W$ state $\mathcal{U}_M(\ket{1}) \propto \ket{10\dots 00} + \ket{01\dots00} + \ket{00\dots0 1}$ and the cat-basis GHZ state $\mathcal{U}_M(\iketCat[1][\sqrt{M}\gamma]) \propto \ket{\gamma}^{\otimes M} + \ket{-\gamma}^{\otimes M}$ can be detected by \cref{thm:GMN-reduced-state} and its associated witnesses for any finite number of modes $M$.

\subsection{GME beyond vacuum interference}
As before, the witness is not limited only to GME states generated through interference with the vacuum.
For the tripartite case, the smoothed Wigner functions of the example states $\iket{D_2^3}$, $\ket{\psi_1}$, and $\ket{\psi_2}$ studied previously in \cref{sec:beyond-vacuum-1} are
\begin{equation}
\begin{aligned}
    &\widetilde{W}_{\tr_{-}\iket{D_2^3}}(\alpha;\{\ket{1}\}) \\
    &\quad{}={} \frac{e^{-\frac{3}{2}\abs{\alpha}^2}}{32\pi}\pqty{
        81\abs{\alpha}^6  - 234\abs{\alpha}^4 + 216\abs{\alpha}^2 - 16
    } \\
    &\widetilde{W}_{\tr_{-}\ket{\psi_1}}(\alpha;\{\ket{0}\}) \\
    &\quad{}={} \frac{e^{-\frac{3}{2}\abs{\alpha}^2}}{1024\pi}\bigg(
        896 - 216\sqrt{23} + 81\abs{\alpha}^4\pqty\Big{
            16 + 3\sqrt{23}
        } \\
    &\qquad\qquad\qquad{}+{} 12\sqrt{2}\Re[\alpha]
    \pqty{16 + 3\sqrt{23}}
    \pqty{9\abs{\alpha}^2-4}
    \bigg) \\
    &\widetilde{W}_{\tr_{-}\ket{\psi_2}}(\alpha;\{\ket{0}\}) = \frac{e^{-\frac{3}{2}\abs{\alpha}^2}}{64\pi}\pqty{
        243\abs{\alpha}^4 - 144\abs{\alpha}^2 + 8
    }.
\end{aligned}
\end{equation}
Their quantities at particular phase-space points are $\widetilde{W}_{\tr_{-}\iket{D_2^3}}(0;\{\ket{1}\}) = -(2\pi)^{-1}$, $\widetilde{W}_{\tr_{-}\ket{\psi_1}}(0.3;\{\ket{0}\}) = -
0.17$, and $\widetilde{W}_{\tr_{-}\ket{\psi_2}}(0.5;\{\ket{0}\}) = -0.04$, which are all negative.
Therefore, these states are all also detected by \cref{thm:GMN-reduced-state}.

Further notice that we used $\ket{1}$ when specifying the smoothing kernel in defining $\widetilde{W}_{\tr_{-}\iket{D_2^3}}$.
If we instead used $\ket{0}$, we would have
\begin{equation}
    \widetilde{W}_{\tr_{-}\iket{D_2^3}}(\alpha;\{\ket{0}\}) = \frac{e^{-\frac{3}{2}\abs{\alpha}^2}}{24\pi}\pqty{
        8 + \pqty{9\abs{\alpha}^2 - 4}^2
    },
\end{equation}
which is a positive function.
This demonstrates the necessity in some cases to use a non-Gaussian kernel function $K(\alpha;\mathcal{R})$ for the GME of the state to be detected.

As another example, take the family of tripartite $N00N$ states, which are defined for $N \geq 1$ as
\begin{equation}
    \ket{\nu_N^3} \coloneqq \frac{1}{\sqrt{3}}\pqty{
        \ket{N00} + \ket{0N0} + \ket{00N}
    }.
\end{equation}
They are multipartite generalizations of the more familiar $\ket{N00N} \propto \ket{N0}+\ket{0N}$, and also share the same features as their bipartite counterparts, like their resourcefulness in metrological tasks like quantum phase estimation \cite{multipartite-N00N}.

The smoothed Wigner functions for the first few tripartite $N00N$ states are
\begin{align}
    &\begin{aligned}
        &\widetilde{W}_{\tr_{-}\iket{\nu_2^3}}(\alpha;\{\ket{1}\}) \\
        &\qquad{}={} \frac{e^{-\frac{3}{2}\abs{\alpha}^2}}{64\pi}\pqty{
            16 + 264\abs{\alpha}^2 - 234\abs{\alpha}^4 +81\abs{\alpha}^6
        },
    \end{aligned}\\
    &\begin{aligned}
        &\widetilde{W}_{\tr_{-}\iket{\nu_3^3}}(\alpha;\{\ket{0}\}) \\
        &\qquad{}={} \frac{e^{-\frac{3}{2}\abs{\alpha}^2}}{64\pi}\pqty{
            {-16} + 216\abs{\alpha}^2 - 54\abs{\alpha}^4 + 27\abs{\alpha}^6
        },
    \end{aligned}\\
    &\begin{aligned}
        &\widetilde{W}_{\tr_{-}\iket{\nu_4^3}}(\alpha;\{\ket{1}\}) \\
        &\qquad{}={} \frac{e^{-\frac{3}{2}\abs{\alpha}^2}}{4096\pi}\Big(
            {-256} + 16512\abs{\alpha}^2 - 21312\abs{\alpha}^4 \\
            &\qquad\qquad\qquad\qquad{}+{} 12384\abs{\alpha}^6 - 1998\abs{\alpha}^8 + 243\abs{\alpha}^{10}
        \Big),
    \end{aligned}\\
    &\begin{aligned}
        &\widetilde{W}_{\tr_{-}\iket{\nu_5^3}}(\alpha;\{\ket{0}\}) \\
        &\qquad{}={} \frac{e^{-\frac{3}{2}\abs{\alpha}^2}}{20480\pi}\Big(
        {-1280} + 28800\abs{\alpha}^2 - 14400\abs{\alpha}^4 \\
            &\qquad\qquad\qquad\qquad{}+{} 21600\abs{\alpha}^6 - 1350\abs{\alpha}^8 + 243\abs{\alpha}^{10}
        \Big).
    \end{aligned}
\end{align}
Here we find that $\widetilde{W}_{\tr_{-}\iket{\nu_N^3}}(\alpha;\iket{(1+(-1)^N)/2})$ is negative for $3 \leq N \leq 5$ and that the negativity is present at the origin $\alpha = 0$.
Extending this intuition to larger $N$, we compute in Appendix~\ref{apd:tripartite-N00N} that
\begin{equation}
    \widetilde{W}_{\tr_{-}\iket{\nu_N^3}}\pqty{
        0;\{\ket{\scriptscriptstyle\frac{1+(-1)^N}{2}}\}
    } = -\frac{2^{-N}}{\pi} \times \begin{cases}
        2 & \text{if $N$ odd,}\\
        N-3\!\! & \text{otherwise.}
    \end{cases}
\end{equation}
Therefore, all tripartite $N00N$ states with $N \neq 2$ can be detected by \cref{thm:GMN-reduced-state}.

Finally, we note that we were unable to find an appropriate filter function that detects Dicke states with \cref{thm:GMN-reduced-state} for $M>3$ modes.
Instead, this witness can detect other classes of states that involve superpositions of higher-order Dicke and $N00N$ states of the form
\begin{equation}
\begin{aligned}
    \ket{\{n_1 n_2 \cdots n_M\}} &\propto \ket{n_1n_2\cdots n_M} + \ket{n_2n_1\cdots n_M} \\
    &\qquad{}+{} \text{\emph{all other permutations}}.
\end{aligned}
\end{equation}
For example, for the $M=4,5$ states
\begin{equation}
\begin{aligned}
    \ket{\psi_4} &= \frac{1}{\sqrt{2}}\pqty{
        \ket{\{2100\}} + \ket{\{1110\}}
    }, \\
    \ket{\psi_5} &= \frac{1}{\sqrt{2}}\pqty{
        \ket{\{21100\}} + \ket{\{11110\}}
    },
\end{aligned}
\end{equation}
their smoothed Wigner functions, defined with a filter function specified by $\mathcal{R} = \{\ket{1}\}_{m=1}^{M-2}$, obtain the values $\widetilde{W}_{\ket{\psi_4}}(0;\mathcal{R}) = -(139\sqrt{3}-144)/512\pi$ and $\widetilde{W}_{\ket{\psi_5}}(0;\mathcal{R}) = -2(11\sqrt{6}-4)/81\pi$.
Therefore, both these states can be detected using \cref{thm:GMN-reduced-state}.

\subsection{Robustness in the presence of noise}
Consider again the damping channel $\mathcal{D}_\eta$ from before.
By the same proof as \cref{prop:center-of-mass-nonclassicality-depth}, it can be shown that the nonclassicality depth of a lossy state is $\tau_c(\tr_{-}[\mathcal{D}_\eta(\rho)]) = (1-\eta)\tau_c(\tr_{-}\rho)$.
Hence, the GME of $\rho$ can be robustly detected in the presence of loss when $(1-M^{-1})(1-\eta)^{-1} < \tau_c(\tr_{-}\rho) \leq 1$.
This means that the threshold amount of tolerable noise decreases with the number of modes as $\eta < M^{-1}$, which is a scaling that can be challenging for noisy systems with large number of modes.
Nonetheless, it still shows that there are noisy GME states that can be detected by this criterion for all finite $M$.

Let us take a look at some explicit examples.
The Wigner function of a lossy $M$-mode $W$ state reads
\begin{equation}
    W_{\rhoW}(\vec{\alpha}) =
    \frac{2^M e^{-2\abs{\vec{\alpha}}^2}}{\pi^M}\pqty{\frac{
        4(1-\eta)
        |\vec{1}^T\vec{\alpha}
    |^2}{M} + 2\eta-1}.
\end{equation}
Now, to apply \cref{thm:GMN-reduced-state}, choose the convolution kernel $K(\alpha;\mathcal{R})$ such that $\mathcal{R}= \{\varrho_m = \ket{0}\}_{m=1}^{M-2}$, with which we obtain the smoothed Wigner function
\begin{equation}
\begin{aligned}
    \widetilde{W}_{\tr_{-}\rhoW}(\alpha) &=
    \frac{4e^{-2\abs{\alpha}^2}}{\pi(M-1)}\pqty{
        M(1-\eta)\abs{\alpha}^2 + \frac{M\eta-1}{2}
    }.
\end{aligned}
\end{equation}
Substituting this into \cref{thm:GMN-reduced-state} with $\alpha = 0$ gives $(1-M\eta)/(M-1) \leq \min_{\sigma\notin\operatorname{GME}}\| \sigma-\mathcal{D}_\eta(\ketW) \|_1$.
As expected, the GME of a lossy $M$-mode $W$ state can be detected with this witness for $\eta < 1/M$, although the magnitude of violation decreases with $\eta$.

Comparing this to the previous witness, we find that \cref{col:GME-Wigner-witness} detects $\rhoW[3]$ with loss $\eta_{\max} \approx 0.36$ which is more robust than the threshold $\eta_{\max} \approx 0.33$ for \cref{thm:GMN-reduced-state}.
Meanwhile, the converse is true for $3 < M \leq 6$ modes: Notably, the noise threshold for $\rhoW[6]$ is $\eta_{\max} \approx 0.16$ for \cref{thm:GMN-reduced-state}, which is more than double the threshold of $\eta_{\max} \approx 0.07$ for \cref{col:GME-Wigner-witness}.
Most drastically, the previous witnesses fail to detect any $W$ states for $M > 6$, while the current witness can detect them for any number of modes.

As noted above, the witnesses in Secs.~\ref{wit:BS-A} and \ref{wit:BS-Rand} are direct implementations of Theorem~\ref{thm:GMN-reduced-state}.
Therefore, any GME that can be detected by the theorem, with the restriction that the kernel is a Gaussian kernel for the witness in \ref{wit:BS-Rand}, can be directly detected by those witnesses.

\subsection{With measurement of finite characteristic function points}
If only characteristic function measurements can be implemented in the experimental setup, then only the witness in Sec.~\ref{wit:D} can be used.
In this section, we shall study the effectiveness of the witness in Sec.~\ref{wit:D} for detecting the GME of noisy W and entangled cat states.

Let us start with the noisy $M$-mode $W$ state.
By choosing the vacuum in specifying the kernel matrix, and parameterizing the phase-space points as $\Xi = \{\xi_n/\sqrt{2(M-1)}\}_{n=1}^N$, we have
\begin{equation}
\begin{aligned}
    &\bqty{
        \mathbf{C}\pqty{\rhoW;\Xi}
        \circ
        \mathbf{K}\pqty{\{\ket{0}\}_{m=1}^{M-2};\Xi}
    }_{n,n'} \\
    &= \frac{e^{-\frac{M\abs{\xi_n-\xi_{n'}}^2}{4(M-1)}}}{N}\pqty{
        1-(1-\eta)\tfrac{M\abs{\xi_n-\xi_{n'}}^2}{2(M-1)}
    } e^{-\frac{(M-2)\abs{\xi_n-\xi_{n'}}^2}{4(M-1)}} \\
    &= \frac{e^{-\frac{\abs{\xi_n-\xi_{n'}}^2}{2}}}{N}\Bqty{
        1 -\bqty{
            1 - \pqty{\tfrac{1+\eta}{2} - \tfrac{1-\eta}{2(M-1)}}
        }\abs{\xi_n-\xi_{n'}}^2
    },
\end{aligned}
\end{equation}
whereupon we find that $\left\|\mathbf{C}(\rhoW)\circ\mathbf{K}\right\|_1$, when maximized over $\Xi$, depends only on the parameter $\zeta(M,\eta) \coloneqq (1+\eta)/2 - (1-\eta)(M-1)^{-1}/2$.
We use this parametrization as it can be further shown that $[\mathbf{C}(\rhoW)\circ\mathbf{K}]_{n,n'} = \chi_{\rho_{\zeta}}(\xi_n-\xi_{n'})$, where we take the right-hand side to be the characteristic function of the effective single-mode state $\rho_{\zeta} = (1-\zeta)\ketbra{1} + \zeta\ketbra{0}$.
This gives $\zeta$ an interpretation as an effective noise parameter, and since $\left\|\mathbf{C}(\rhoW)\circ\mathbf{K}\right\|_1 \leq \int_{\mathbb{C}}\dd[2]{\alpha} |W_{\rho_{\zeta}}(\alpha)|$, it further informs us that $\zeta(M,\eta) < 1/2$ is a necessary condition for the GME of $\rhoW$ to be detected with the witness in Sec.~\ref{wit:D} for this choice of matrix kernel.

In implementing this witness, we numerically maximize $\left\|\mathbf{C}(\rhoW)\circ\mathbf{K}\right\|_1$ over $\Xi$ for a fixed dimension $N$ of $\Xi$.
This is plotted against the effective noise parameter $\zeta(M,\eta)$ in \cref{fig:W_state}(a).
We find that the matrix dimension $N \geq 6$ is needed to detect a tripartite $W$ state, which requires more measurement settings than the witness in Sec.~\ref{wit:D-BS}.
On the other hand, $N \geq 19$ is needed to detect an $M=7$ partite $W$ state, whereas the witness in Sec.~\ref{wit:D-BS} cannot detect $\rhoW[7]$ at all.

\begin{figure}
    \centering
    \includegraphics[width=\columnwidth]{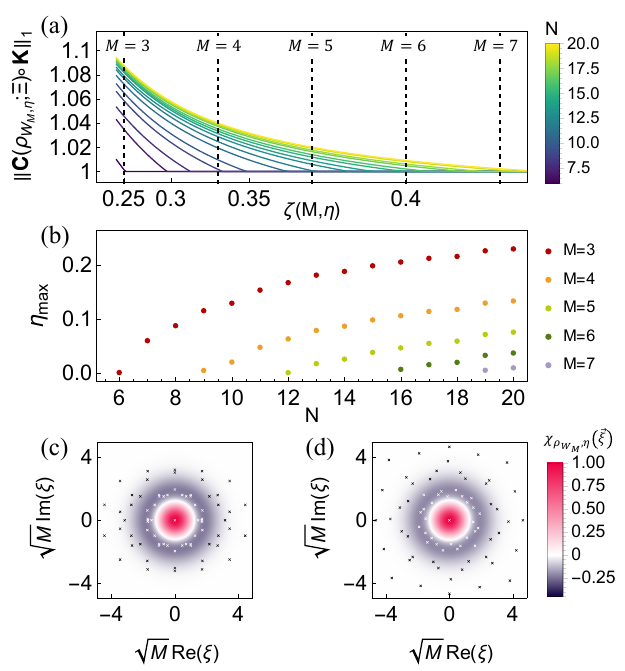}
    \caption{
        \label{fig:W_state}(a) Implementation of the witness in Sec.~\ref{wit:D} for noisy $W$ states, which only require finite number of measurement settings.
        For any $M$ and $\eta$, the quantity $\left\|\mathbf{C}(\rhoW)\circ\mathbf{K}\right\|_1$ depends only on the parameter $2\zeta(M,\eta) \coloneqq (1+\eta) - (1-\eta)/(M-1)$ used in the horizontal axis.
        GME is detected when $\|\mathbf{C}\circ\mathbf{K}\|_1 > 1$ for the given $\eta$ and $M$, while the dashed vertical lines are the GME thresholds for $\eta = 0$ and $M \in \{3,4,5,6,7\}$.
        (b) The maximum loss $\eta_{\max}$ such that the GME of $\rhoW$ can still be detected, plotted against $N$ for different values of $M$.
        While $N$ is the matrix dimension and thus indicative of the number of measurement settings required, the actual number of measurement settings tend to be fewer due to duplicates.
        For example, the detection of both (c) $M=4$ and (d) $M=5$
        both require measurements at 36 points, although the former requires the construction of a $9\times9$ matrix while the latter requires a $12\times12$ matrix.
    }
\end{figure}

In order to study the robustness of the witness more closely, we also plot in \cref{fig:W_state}(b) the maximum noise parameter $\eta_{\max}$, such that $\rhoW$ can still be detected by this witness with an $N$-dimensional matrix.
As expected, making more measurements allows us to detect noisier states, and states with fewer modes are more robust against loss.
With $N\sim 15$, which corresponds to measurements of $\sim 100$ phase-space points, the witness can detect tripartite $W$ states up to $\eta \lesssim 0.2$ and four-partite $W$ states up to $\eta \lesssim 0.1$.

Note however that the actual number of measured points needed can also be fewer than the size of matrix dimension in general.
Take the example of $\rhoW[4]$ and $\rhoW[5]$, as shown in Refs.~\ref{fig:W_state}(c) to \ref{fig:W_state}(d).
While we need $N=9$ to detect $\rhoW[4]$ and $N=12$ to detect $\rhoW[5]$, both require measurements at just 37 points in phase space.

Finally, we highlight a unique feature of this witness not present in previous ones.
In its implementation, let us assume that the matrix elements of $\mathbf{C}$ has already been obtained experimentally.
Then, one still has the freedom to choose the matrix kernel $\mathbf{K}(\{\varrho_m\}_{m=1}^{M-2})$ when computing $\left\|\mathbf{C}\circ\mathbf{K}\right\|_1$.
Hence, we can still optimize $\max_{\mathbf{K}}\left\|\mathbf{C}\circ\mathbf{K}\right\|_1$ over the matrix kernels with little experimental overhead, as this only requires classical postprocessing without needing more measurements.

Take the tripartite entangled cat state, for which
\begin{equation}
\begin{aligned}
    &[\mathbf{C}(\ketCat[3];\Xi)]_{n,n'} = \frac{e^{-\frac{3\abs{\xi_n-\xi_{n'}}^2}{2}}}{1+e^{-6\abs{\gamma}^2}}\bigg(\\
        &\qquad\cos(6\Im[(\xi_n-\xi_{n'})\gamma^*]) \\
        &\qquad{}+{}
        e^{-6\abs{\gamma}^2}\cosh(6\Re[(\xi_n-\xi_{n'})\gamma^*])
    \bigg),
\end{aligned}
\end{equation}
and imagine the scenario where the system has been measured and $\mathbf{C}(\ketCat[3];\Xi)$ already constructed for fixed points $\Xi = \{\pm\xi_0,\pm\xi_0^*,\pm \Re[\xi_0]\}$ with $\xi_0 = 10/11 + i7/17$.

Then, we shall optimize $\left\|\mathbf{C}\circ\mathbf{K}_s\right\|_1$ over a family $\mathbf{K}_s(\Xi) \coloneqq \mathbf{K}(\{\ket{s}_{\text{sq}}\}_{m=1}^{M-2};\Xi)$ of kernel matrices, where $\ket{s}_{\text{sq}} \coloneqq e^{\ln(s)(a^{2} - a^{\dag2})/2}\ket{0}$ is the squeezed vacuum and $s$ is a squeezing parameter, with $s=1$ corresponding to the vacuum case.
The elements of the matrix kernel are
\begin{equation}
    [\mathbf{K}_s(\Xi)]_{n,n'} =
        e^{-\frac{s^2(M-2)}{2}\Re[\xi_n-\xi_{n'}]^2}
        e^{-\frac{M-2}{2s^2}\Im[\xi_n-\xi_{n'}]^2}.
\end{equation}
The computed $\left\|\mathbf{C}\circ\mathbf{K}_s\right\|_1$ is plotted against $s$ in \cref{fig:witD_cat}.

\begin{figure}
    \centering
    \includegraphics[width=80mm]{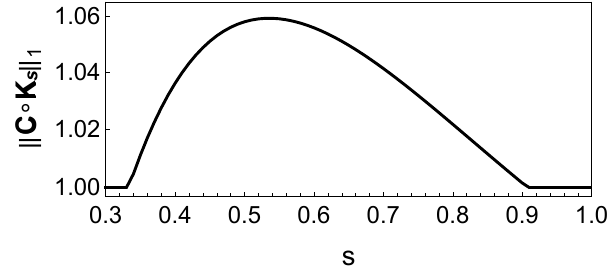}
    \caption{
        \label{fig:witD_cat} The trace norm of $\mathbf{C}\circ\mathbf{K}_s$ for fixed characteristic function measurements at phase-space points $\Xi$, plotted for different choices of the kernel matrix $\mathbf{K}_s$.
        Here, $\mathbf{K}_s$ is related to the characteristic function of the squeezed state with squeezing parameter $s$, where $s=1$ corresponds to the vacuum state $\ket{0}$.
        For the chosen $\Xi$, the tripartite entangled cat with cat size $|\gamma|=1$ cannot be detected with the vacuum matrix kernel $\mathbf{K}_{s=1}$.
        The maximum violation is $\left\|\mathbf{C}\circ\mathbf{K}_{s=0.535}\right\|_1 = 1.059$.
    }
\end{figure}

We find that, given the fixed choice of $\Xi$, the GME of $\ketCat[3]$ can only be detected when $0.33 \leq s \leq 0.91$, with $\max_s \left\|\mathbf{C}\circ\mathbf{K}_s\right\|_1 = 1.059$ at $s=0.535$.
In particular, GME would not have been detected if vacuum states were chosen in defining the kernel matrix.
This shows that the detection of GME with this technique can be improved at no further experimental cost, by simply post-processing the measured data.

\section{Conclusion}

In this work, we built upon the two theorems proven in our companion Letter~\cite{CV-GME-letter} that relate Wigner function negativity to GME detection, and developed approaches for their experimental implementation with controlled Gaussian unitaries.
In particular, we showed that controlled parity, displacement, and beam-splitter operations, either individually performed or in combination with each other, allow the GME of a variety of relevant multipartite quantum states to be detected.

To benchmark our GME criteria, we analyzed their applicability to paradigmatic families of multipartite entangled states that are of realistic and relevant experimental realization.
These include multimode W and cat states, as well as two-excitation Dicke states and tripartite states without permutation symmetry.
Meanwhile, to assess the robustness of our witnesses against imperfections, we examined the effect of energy relaxation on the capability of our witness to detect GME.
Last, considering realistic experimental implementations, we also studied the effects of finite detection resolution and discrete phase-space sampling in performing Wigner function measurements.

In summary, we have presented a number of approaches to detect GME in continuous-variable systems through the measurement of the Wigner or characteristic function over only a finite region or finite number of points in phase space.
These methods pave the way for certifying GME in a variety of experimental platforms where conventional criteria based on quadrature measurements can be difficult to implement.
Examples of such platforms include cavity and circuit quantum electrodynamics, circuit quantum acoustodynamics, as well as trapped ions and atoms, where bosonic degrees of freedom are not accessed directly, but rather read out via a coupled two-level qubit.
Establishing GME in such platforms is a crucial step toward benchmarking the quantum resources needed for a range of quantum information tasks, including quantum computing, metrology, communication, and sensing.

Last, let us mention some aspects of our witnesses that are beyond the scope of the current work, but could inspire further research.
One is the optimal choice of the 2D slice upon which to measure the Wigner or characteristic function when implementing the witnesses, i.e., the best choice of the coefficients $\vec{y}$ and $\vec{z}$ with $\vec{y}\circ\vec{y}^* - \vec{z}\circ\vec{z}^* = \vec{1}$ that can detect GME of the studied state.
For states generated via vacuum interference, \cref{prop:center-of-mass-absolute-volume,prop:center-of-mass-nonclassicality-depth} show that $\vec{y} = \vec{1}$ and $\vec{z} = \vec{0}$ is a suitable choice.
Meanwhile, our worked examples of the Dicke and $N00N$ states suggest that a good rule of thumb is to choose $\vec{y} = \vec{1}$ and $\vec{z} = \vec{0}$ for states that are invariant under permutation of the modes.
However, whether this choice of 2D slice is optimal for detecting the GME of the above states, and more generally, the strategy for choosing the most optimal 2D slice for an arbitrary state, remains an open question.

Another aspect is the possibility of constructing stronger continuous-variable entanglement witnesses using similar techniques.
While we showed that GME can be detected using our criteria for any finite number of modes $M$, we also found that the robustness decreases with $M$, which means that GME becomes harder to detect for states with a large number of modes.
This could be due to the fact that all of our criteria detect GME by assessing only $2$ of the full $2M$ phase-space dimensions, which constitutes a vanishingly small proportion of the full information available about an $M$-mode state.
Extending the witness by incorporating more information about the state can only improve the detection of GME, and could therefore improve robustness for systems with more modes.
Such extensions, like analogous relationships between GME and negativity in a $2M'$-dimensional slice of the Wigner function for $1 < M' < M$, or the trade-off between robustness and scalability in terms of the required measurements, are left to future works.

\begin{acknowledgments}
This work is supported by the National Research Foundation, Singapore, under its Centre for Quantum Technologies Funding Initiative (S24Q2d0009), the National Natural Science Foundation of China (Grants No. 12125402, No. 12534016, No. 12405005, No. 12505010), the Beijing Natural Science Foundation (Grant No. Z240007), and the Quantum Science and Technology-National Science and Technology Major Project (No. 2024ZD0302401, No. 2021ZD0301500).
J.G. acknowledges the support of the Postdoctoral Fellowship Program of CPSF (No. GZB20240027), and the China Postdoctoral Science Foundation (No. 2024M760072).
S.L. acknowledges the China Postdoctoral Science Foundation (No. 2023M740119).
M.F. was supported by the Swiss National Science Foundation Ambizione Grant No. 208886, and by The Branco Weiss Fellowship -- Society in Science, administered by the ETH Z\"{u}rich.
\end{acknowledgments}

\section*{Data Availability}
The data that support the findings of this article are openly available \cite{data}.

\bibliography{refs}

\appendix

\begin{widetext}
\section{\label{apd:proofs}Proofs of Propositions}
\subsection{\label{apd:bochner-extended}Proof of Proposition~\ref{prop:bochner-extended}}
\begin{proof}[Proof of Proposition~\ref{prop:bochner-extended}]
    Let us consider two scenarios: First, assume that $\int_{\mathbb{C}^M}\dd[2M]{\vec{\alpha}} \tr[ R \Pi(\vec{\alpha})]$ does not converge absolutely.
    Then, the proposition is trivially true since the right-hand side of the inequality will be infinite, while the left-hand side is the trace of a finite-dimensional matrix, so $\|\mathbf{C}\|_1 < \infty$.

    Otherwise, $\int_{\mathbb{C}^M}\dd[2M]{\vec{\alpha}} |\tr[R \Pi(\vec{\alpha})]|$ takes a finite value.
    Then, we start with the usual Wigner function construction that for any positive trace-class operator $\rho \succeq 0$, ${\tr}[\rho D(\vec{\xi})] = (2/\pi)^M\int_{\mathbb{C}^M}\dd[2M]{\vec{\alpha}} e^{\vec{\xi}\wedge\vec{\alpha}} \tr[ \rho \Pi(\vec{\alpha})]$.
    This can be extended to any trace-class operator $R$. First, it being trace class implies that we can decompose $R = R_+ - R_-$ with $R_\pm \succeq 0$ and
    $\left\|R_\pm\right\|_1 = \tr(R_\pm) < \infty$.
    Then, since $\Pi(\alpha)$ is unitary, $\Pi(\alpha)R$ is itself trace class, and thus the difference $\tr[\Pi(\alpha)R_+-\Pi(\alpha)R_-] = \tr[\Pi(\alpha)R_+] - \tr[\Pi(\alpha)R_-]$ is linear in the trace.
    Next, by Fubini's theorem, the absolute convergence of the integral of $e^{\vec{\xi}\wedge\vec{\alpha}}\pqty{\tr[\Pi(\alpha)R_+] - \tr[\Pi(\alpha)R_-]}$ allows us to split the integral into a difference of two integrals, use the result for positive trace-class operators, then combine the difference using the linearity of the trace of the trace-class operators $D(\vec{\xi})R_\pm$.

    The end point of the above discussion is the relation $
        {\tr}[R D(\vec{\xi})] = (2/\pi)^M
        \int_{\mathbb{C}^M}\dd[2M]{\vec{\alpha}}
        e^{\vec{\xi}\wedge\vec{\alpha}}
        \tr[ R \Pi(\vec{\alpha})]
    $ between the matrix elements of $\mathbf{C}(R;\xi)$ and their Fourier transform.
    Then, by writing $
        \|\mathbf{C}(R;\Xi)\|_1 =
        \tr[\mathbf{C}(R;\Xi)\, P_+] -
        \tr[\mathbf{C}(R;\Xi)\, P_-]$,
    where $P_{\pm}$ is the projector onto the positive and negative eigenspace of $\mathbf{C}(R;\Xi)$, respectively, we have
    \begin{equation}
    \begin{aligned}
        \|\mathbf{C}(R;\Xi)\|_1 &=
        \pqty{\frac{2}{\pi}}^M\int_{\mathbb{C}^M}\dd[2M]{\vec{\alpha}}
        \overbrace{
            \frac{1}{N}\sum_{n,n'=1}^N
                e^{\vec{\xi}_n\wedge\vec{\alpha}}
                [P_+]_{n,n'}
                e^{-\vec{\xi}_{n'}\wedge\vec{\alpha}}
        }^{\eqqcolon \vec{u}_{\vec{\alpha}}^\dag P_+ \vec{u}_{\vec{\alpha}}}
        \tr[R\Pi(\vec{\alpha})] \\
        &\qquad{}-{} \pqty{\frac{2}{\pi}}^M\int_{\mathbb{C}^M}\dd[2M]{\vec{\alpha}}
        \underbrace{
            \frac{1}{N}\sum_{n,n'=1}^N
                e^{\vec{\xi}_n\wedge\vec{\alpha}}
                [P_-]_{n,n'}
                e^{-\vec{\xi}_{n'}\wedge\vec{\alpha}}
        }_{\eqqcolon \vec{u}_{\vec{\alpha}}^\dag P_- \vec{u}_{\vec{\alpha}}}
        \tr[R\Pi(\vec{\alpha})],
    \end{aligned}
    \end{equation}
    where $\vec{u}_{\vec{\alpha}} \coloneqq (
        e^{-\vec{\xi}_{1}\wedge\vec{\alpha}},
        e^{-\vec{\xi}_{2}\wedge\vec{\alpha}},
        \dots,
        e^{-\vec{\xi}_{N}\wedge\vec{\alpha}}
    )/\sqrt{N}$.
    Since $\vec{u}_{\vec{\alpha}}$ are unit vectors and $P_{\pm}$ are projectors, we have $0 \leq \vec{u}_{\vec{\alpha}}^\dagger P_{\pm} \vec{u}_{\vec{\alpha}} \leq 1$, and as such
    \begin{align*}
        \|\mathbf{C}(R;\Xi)\|_1 &=
        \pqty{\frac{2}{\pi}}^M\int_{\mathbb{C}^M}
        \dd[2M]{\vec{\alpha}} \;
        \vec{u}_{\vec{\alpha}}^\dag P_+ \vec{u}_{\vec{\alpha}} \;
        \tr[R\Pi(\vec{\alpha})] \\
        &\qquad{}-{} \pqty{\frac{2}{\pi}}^M\int_{\mathbb{C}^M}
        \dd[2M]{\vec{\alpha}} \;
        \vec{u}_{\vec{\alpha}}^\dag P_- \vec{u}_{\vec{\alpha}} \;
        \tr[R\Pi(\vec{\alpha})] \\
        &\leq \pqty{\frac{2}{\pi}}^M\int_{\mathbb{C}^M}
        \dd[2M]{\vec{\alpha}} \;
        \vec{u}_{\vec{\alpha}}^\dag P_+ \vec{u}_{\vec{\alpha}} \;
        \begin{cases}
            \tr[R\Pi(\vec{\alpha})] & \text{if $\tr[R\Pi(\vec{\alpha})] \geq 0$,} \\
            0 & \text{otherwise,}
        \end{cases} \\
        &\qquad{}-{} \pqty{\frac{2}{\pi}}^M\int_{\mathbb{C}^M}
        \dd[2M]{\vec{\alpha}} \;
        \vec{u}_{\vec{\alpha}}^\dag P_- \vec{u}_{\vec{\alpha}} \;
        \begin{cases}
            0 & \text{if $\tr[R\Pi(\vec{\alpha})] \geq 0$,} \\
            \tr[R\Pi(\vec{\alpha})] & \text{otherwise,}
        \end{cases} \\
        &\leq \pqty{\frac{2}{\pi}}^M\int_{\mathbb{C}^M}
        \dd[2M]{\vec{\alpha}} \;
        \begin{cases}
            \tr[R\Pi(\vec{\alpha})] & \text{if $\tr[R\Pi(\vec{\alpha})] \geq 0$,} \\
            0 & \text{otherwise,}
        \end{cases} \\
        &\qquad{}-{} \pqty{\frac{2}{\pi}}^M\int_{\mathbb{C}^M}
        \dd[2M]{\vec{\alpha}} \;
        \begin{cases}
            0 & \text{if $\tr[R\Pi(\vec{\alpha})] \geq 0$,} \\
            \tr[R\Pi(\vec{\alpha})] & \text{otherwise,}
        \end{cases} \\
        &= \pqty{\frac{2}{\pi}}^M\int_{\mathbb{C}^M}
        \dd[2M]{\vec{\alpha}} \abs\Big{
            \tr[R\Pi(\vec{\alpha})]
        }.
    \end{align*}
\end{proof}

\subsection{\label{apd:center-of-mass-parity-is-beam splitter}Proof of Proposition~\ref{prop:center-of-mass-parity-is-beam splitter}}
\begin{proof}[Proof of Proposition~\ref{prop:center-of-mass-parity-is-beam splitter}]
    In terms of the collective modes $\{a_{+}\} \cup \{a_{-m}\}_{m=2}^{M}$, each local mode can be written as $a_m = a_+/\sqrt{M} + \sum_{m'=2}^M c_{m,m'} a_{-m'}$. Then,
    \begin{equation}
    \begin{aligned}
        \Pi_{+M} a_m \Pi_{+M} &= \frac{\Pi_{+M}a_+\Pi_{+M}}{\sqrt{M}} + \sum_{m'=2}^M c_{m,m'} \Pi_{+M} a_{-m'} \Pi_{+M}
        = - \frac{a_+}{\sqrt{M}} + \sum_{m'=2}^M c_{m,m'}  a_{-m'} \\
        &= - 2\frac{a_+}{\sqrt{M}} + \underbrace{\pqty{\frac{a_+}{\sqrt{M}} + \sum_{m'=2}^M c_{m,m'}  a_{-m'}}}_{=a_m}
        = a_m - \frac{2}{M}\sum_{m'=1}^M a_{m'} \\
        &= \pqty{1-\frac{2}{M}}a_m - \frac{2}{M}\sum_{m'\neq m} a_{m'}.
    \end{aligned}
    \end{equation}
\end{proof}

\subsection{\label{apd:center-of-mass-absolute-volume}Proof of \cref{prop:center-of-mass-absolute-volume}}
\begin{proof}[Proof of \cref{prop:center-of-mass-absolute-volume}]
First, note that because $\sqrt{M}a' \coloneqq \sum_{m=1}^M U_M^\dag a_m U_M$ satisfies $[a',a^{\prime\dagger}] = \mathbbm{1}$, and we also have that $a' = a_1 + \sum_{m=1}^M\sum_{m'=2}^M c_{m,m'}a_{m'}$, it must be that $\sum_{m=1}^M U_M^\dag a_m U_M = \sqrt{M}a_1$ for the commutation relation to hold.
From this, we can find that
\begin{equation}
    U_M^\dag D(\alpha\vec{1}) U_M =
    U_M^\dag e^{(\alpha\vec{1})\wedge\vec{a}} U_M =
    U_M^\dag e^{\alpha\wedge(\sum_{m=1}^M a_m)} U_M =
    e^{\alpha\wedge(\sum_{m=1}^M U_M^\dag a_m U_M)}=
    e^{\sqrt{M}\alpha\wedge a_1} =
    D_1(\sqrt{M}\alpha).
\end{equation}
Then, since beamsplitting operations preserve the total energy $\vec{a}^\dag\vec{a}$, we also have that $U_M^\dag \Pi U_M = \Pi$.
Putting everything together,
\begin{equation}
\begin{aligned}
    \mathcal{V}_{2D}(\mathcal{U}_M(\rho_1);\mathbb{C})
    &= \pqty{\frac{\pi}{2}}^{M-1}
    \int_{\mathbb{C}}\dd[2]{\alpha} \abs{W_{\mathcal{U}_M(\rho_1)}(\alpha\vec{1})} \\
    &= \frac{2}{\pi}\int_{\mathbb{C}}\dd[2]{\alpha}\abs{
        \tr[
            U_M
            \pqty{\rho_1 \otimes \ketbra{0}^{\otimes M-1}}
            U_M^\dag
            D(\alpha\vec{1}) \Pi D^\dag(\alpha\vec{1})
        ]
    } \\
    &= \frac{2}{\pi}\int_{\mathbb{C}}\dd[2]{\alpha}\abs{
        \tr[
            \pqty{\rho_1 \otimes \ketbra{0}^{\otimes M-1}} \;
            U_M^\dag D(\alpha\vec{1}) U_M \;
            U_M^\dag \Pi U_M \;
            U_M^\dag D^\dag(\alpha\vec{1}) U_M
        ]
    } \\
    &= \frac{2}{\pi}\int_{\mathbb{C}}\dd[2]{\alpha}\abs{
        \tr[
            \pqty{\rho_1 \otimes \ketbra{0}^{\otimes M-1}}
            D_1(\sqrt{M}\alpha) \Pi D_1^\dag(\sqrt{M}\alpha)
        ]
    } \\
    &= \frac{2}{\pi}\int_{\mathbb{C}}\dd[2]{\alpha}\abs{
        \tr[\rho_1 \Pi_1(\sqrt{M}\alpha)]
        \prod_{m=2}^M \bra{0}\!\Pi_m\!\ket{0}
    } \\
    &= \int_{\mathbb{C}}\dd[2]{\alpha}\abs{
        W_{\rho_1}(\sqrt{M}\alpha)
    } = \frac{1}{M}\int_{\mathbb{C}}\dd[2]{\alpha}\abs{
        W_{\rho_1}(\alpha)
    }.
\end{aligned}
\end{equation}
\end{proof}

\subsection{\label{apd:rigorous-integral-error}Proof of \cref{prop:rigorous-integral-error}}
\begin{proof}[Proof of \cref{prop:rigorous-integral-error}]
When an upper bound on the total energy of the system ${\sum_{m=1}^M\tr}[ a_m^{\dagger} a_m \rho] \leq E$ is known, it is possible to obtain rigorous inequalities for the discretized integral.
Starting with the mean value theorem,
\begin{equation}
    \left|W_\rho(\vec{\alpha})-W_\rho(\vec{\beta})\right| \leq|\vec{\alpha}-\vec{\beta}| \max_{0 \leq p
        \leq 1}\left|\nabla W_\rho(p \vec{\alpha}+(1-p) \vec{\beta})\right|.
\end{equation}
The partial derivative of the Wigner function with respect to the complex conjugate variable $\alpha_m^*$ takes the form \cite{quantum-noise}
\begin{equation}
\begin{aligned}
    \partial_{\alpha_m^*} W_\rho(\vec{\alpha})
        &= 2 W_{a_m \rho}(\vec{\alpha})-2 \alpha_m W_\rho(\vec{\alpha})
        = 2\left(\frac{2}{\pi}\right)^M \operatorname{tr}\left[\left(a_m-\alpha_m\right) \rho \Pi(\vec{\alpha})\right].
\end{aligned}
\end{equation}
An analogous expression holds for the partial derivative with respect to $\alpha_m$.
Applying the Cauchy-Schwarz inequality and using the fact that $\Pi(\vec{\alpha})^2 = \mathbbm{1}$,
\begin{equation}
\begin{aligned}
    \left|\partial_{\alpha_m^*} W_\rho(\vec{\alpha})\right|^2
        &= 4\pqty{\frac{2}{\pi}}^{2 M}
            \abs\big{
                \tr[
                    \pqty{(a_m-\alpha_m)\sqrt{\rho}}
                    \pqty{\sqrt{\rho} \Pi(\vec{\alpha})}
                ]
            }^2
        \leq 4\pqty{\frac{2}{\pi}}^{2 M} \tr[
            \abs{a_m-\alpha_m}^2 \rho
        ].
\end{aligned}
\end{equation}
Using the triangle inequality for norms, we find ${\tr}[|a_m-\alpha_m|^2 \rho] \leq (\sqrt{\epsilon_m} + |\alpha_m|)^2$, where $\epsilon_m = {\tr}[a_m^\dagger a_m \rho]$ is the mean energy of the $m$th mode.
Since $|\nabla f|^2 = 2\sum_m(|\partial f/\partial \alpha_m|^2 + |\partial f/\partial \alpha_m^*|^2)$, we obtain
\begin{equation}
    \left|\nabla W_\rho(\vec{\alpha})\right|^2 \leq 16\left(\frac{2}{\pi}\right)^{2M}\left(\sqrt{E}+\sqrt{M}\max_{1 \leq m \leq M}|\alpha_m|\right)^2.
\end{equation}
Now, consider the discretization of the finite phase-space region $\omega$ into a square grid with grid spacing $\Delta \times \Delta$.
The integral becomes $
    \int_\omega \dd[2]{\alpha} |W(\vec{\alpha})|
    = \sum_{\{\alpha_0\}} \int_{R_{\Delta}(\alpha_0)} \dd[2]{\alpha} |W(\vec{\alpha})|
$, where $R_\Delta(\alpha_0)$ denotes a square centered at $\alpha_0$ with side length $\Delta$.
Both sides of the integral are well-defined as $R_\Delta(\alpha_0)$ is clearly Lebesgue-measurable, as is the finite disjoint union $\omega = \bigsqcup_{\{\alpha_0\}} R_\Delta(\alpha_0)$.
For each region $R_\Delta(\alpha_0)$, we can establish the lower bound,
\begin{equation}
    \int_{R_{\Delta(\alpha_0)}}
        \dd[2]{\alpha}
        |W(\vec{\alpha})|
    \geq \abs{
        \int_{R_{\Delta}(\alpha_0)}
            \dd[2]{\alpha}
            W(\vec{\alpha})
    }
    \geq \Delta^2\abs{W(\vec{\alpha}_0)} -\int_{R_{\Delta}(\alpha_0)}
        \dd[2]{\alpha}
        \abs{
            W(\vec{\alpha}) -
            W(\vec{\alpha}_0)
        }.
\end{equation}
Then, noting that for every $\alpha \in R_\Delta(\alpha_0)$, $|\vec{\alpha}-\vec{\alpha}_0| \leq \sqrt{M/2}\Delta$ for $\vec{\alpha} = (\alpha,\alpha,\dots,\alpha)$ and $| p \alpha+(1-p) \alpha_0 | \leq |\alpha_0|+\frac{\Delta}{\sqrt{2}}$,
\begin{equation}
\int_{R_{\Delta(\alpha_0)}}
        \dd[2]{\alpha}
        |W(\vec{\alpha})|
    \geq \Delta^2\abs{W(\vec{\alpha}_0)}
        -\pqty{\frac{2}{\pi}}^M\bqty{
            2 \sqrt{2 M} \Delta^3 \pqty{
                \sqrt{E}+\sqrt{M}\abs{\alpha_0}
            } +
            2 M \Delta^4
        }.
\end{equation}
By using the previous steps to bound the error term and the relation $W(\vec{\alpha}_0) = (2/\pi)^M {\tr}[\rho\Pi(\vec{\alpha}_0)]$, we arrive at the following GME witness: If the total energy is bounded by $E$, then
\begin{equation}
    \sum_{\{\alpha_0\}}\bqty{
        \Delta^2\abs{
            \tr[\rho \Pi(\vec{\alpha}_0)]
        } - 2\sqrt{2M}\Delta^3\pqty{
            \sqrt{E} + \sqrt{M}|\alpha_0|
        } - 2M\Delta^4
    } > \frac{\pi}{4\sqrt{M-1}}
\end{equation}
implies that $\rho$ is GME.
\end{proof}

\subsection{\label{apd:center-of-mass-nonclassicality-depth}Proof of \cref{prop:center-of-mass-nonclassicality-depth}}
\begin{proof}[Proof of \cref{prop:center-of-mass-nonclassicality-depth}]
First, note that given $\sqrt{M} a_+ \coloneqq \sum_{m=1}^{M} a_m$, we have $U_M^\dag a_+ U_M = a_1$ for the same reason included previously in the proof of \cref{prop:center-of-mass-absolute-volume}.
Then, the Wigner function of $\tr_-[\mathcal{U}_M(\rho_1)]$ is
\begin{equation}
    W_{\tr_-[\mathcal{U}_M(\rho_1)]}(\alpha) =
    \frac{2}{\pi}\tr[ U_M \pqty{\rho_1 \otimes \ketbra{0}^{\otimes M-1}} U_M^\dag \Pi_+(\alpha) ] =
    \frac{2}{\pi}\tr[ \pqty{\rho_1 \Pi_1(\alpha)} \otimes \ketbra{0}^{\otimes M-1}] =
    \frac{2}{\pi}\tr[\rho_1\Pi_1(\alpha)] =
    W_{\rho_1}(\alpha).
\end{equation}
Finally, by the correspondence between Wigner functions and operators in Hilbert space, this must mean that $\tr_-[\mathcal{U}_M(\rho_1)] = \rho_1$ and therefore $\tau_c(\tr_-[\mathcal{U}_M(\rho_1)]) = \tau_c(\rho_1)$.
\end{proof}

\section{\label{apd:tripartite-N00N}Smoothed Wigner function of Tripartite $N00N$ States}
From the proof of \cref{thm:GMN-reduced-state} in our companion Letter \cite{CV-GME-letter}, we know that we can interpret the smoothed Wigner function $\widetilde{W}_{\tr_{-}\rho}(0;\{\varrho_m\}_{m=1}^{M-2}) = (2/\pi){\tr}[\Pi_{+(2M-2)}(\rho\otimes\bigotimes_{m=1}^{M-2}\varrho_m)]$ as the expectation value of the beam splitter $\Pi_{+(2M-2)}$, which has the action $(M-1)\Pi_{+(2M-2)}a_m\Pi_{+(2M-2)} = (M-2)a_m - \sum_{m' \neq m}^{2(M-1)} a_{m'}$.

For $M=3$, $\rho = \ketbra{\nu_N^3}$, and $\varrho_m = \ketbra{0}$, we therefore have
\begin{equation}
\begin{aligned}
    \Pi_{+4}(\ket{\nu_N^3}\otimes\ket{0})
    &= \frac{1}{2^N\sqrt{3\,N!}}\pqty{
        \pqty{ a_1-a_2-a_3-a_4}^N +
        \pqty{-a_1+a_2-a_3-a_4}^N +
        \pqty{-a_1-a_2+a_3-a_4}^N
    }\ket{0000} \\
    &= \frac{(-1)^N}{2^N\sqrt{3\,N!}}
    \sum_{j=0}^N\sum_{k=0}^{j}\sum_{l=0}^{N-j}
    \binom{N}{j}\binom{j}{k}\binom{N-j}{l}
    \pqty{(-1)^{k} + (-1)^{j-k} + (-1)^l}
    a_1^{\dagger(j-k)}
    a_2^{\dagger k}
    a_3^{\dagger l}
    a_4^{\dagger(N-j-l)}
    \ket{0000} \\
    &= \frac{(-1)^N}{2^N\sqrt{3}}
    \sum_{j=0}^N\sum_{k=0}^{j}\sum_{l=0}^{N-j}
    \sqrt{
        \binom{N}{j}
        \binom{j}{k}
        \binom{N-j}{l}
    }
    \pqty{(-1)^{k} + (-1)^{j-k} + (-1)^l}
    \ket{j-k,k,l,N-j-l}.
\end{aligned}
\end{equation}
From this, we can read out that
\begin{equation}
    \widetilde{W}_{\tr_{-}\rho}(0;\{\ket{0}\}) = \frac{
        \bra{N000}\Pi_{+4}(\ket{\nu_N^3}\otimes\ket{0}) +
        \bra{0N00}\Pi_{+4}(\ket{\nu_N^3}\otimes\ket{0}) +
        \bra{00N0}\Pi_{+4}(\ket{\nu_N^3}\otimes\ket{0})
    }{2^{-1}\pi\sqrt{3}} \\
    = \frac{(-1)^N\pqty{2+(-1)^N}}{2^{N-1}\pi}.
\end{equation}
This is sufficient for odd tripartite $N00N$ states to be detected, but the smoothed Wigner function is still positive for even $N$. Instead, by using $\ket{1}$ in defining the kernel function to obtain the smoothed Wigner function $\widetilde{W}_{\tr_{-}\rho}(0;\{\ket{1}\})$,
\begin{equation}
\begin{aligned}
    \Pi_{+4}\pqty{\ket{\nu_N^3}\otimes\ket{1}}
    &= \Pi_{+4}a_4^\dag\pqty{\ket{\nu_N^3}\otimes\ket{0}}
    = \frac{1}{2}\pqty{-a_1^\dag-a_2^\dag-a_3^\dag+a_4^\dag}\Pi_{+4}\pqty{\ket{\nu_N^3}\otimes\ket{0}},
\end{aligned}
\end{equation}
and therefore,
\begin{equation}
\begin{aligned}
    \widetilde{W}_{\tr_{-}\rho}(0;\{\ket{1}\}) &= \frac{2}{\pi}\pqty{\bra{\nu_N^3}\otimes\bra{1}}\Pi_{+4}\pqty{\ket{\nu_N^3}\otimes\ket{1}} \\
    &= -\frac{\sqrt{N}}{\pi}\pqty{\bra{\nu_{N-1}^3}\otimes\bra{1}}\Pi_{+4}\pqty{\ket{\nu_N^3}\otimes\ket{0}} + \frac{1}{\pi}\pqty{\bra{\nu_N^3}\otimes\bra{0}}\Pi_{+4}\pqty{\ket{\nu_N^3}\otimes\ket{0}} \\
    &= -\frac{(-1)^N N\pqty{2-(-1)^N}}{2^N\pi}+\frac{(-1)^N\pqty{2+(-1)^N}}{2^N\pi} = -\frac{2(-1)^{N}(N-1) - (N+1)}{2^N\pi},
\end{aligned}
\end{equation}
which is negative for even $N > 2$, as desired.
\end{widetext}
\end{document}